\documentclass[11pt]{article}
\usepackage[latin1]{inputenc} 

\usepackage[english]{babel}
\usepackage{indentfirst}
\usepackage{fancyhdr}
\usepackage{graphicx}
\usepackage{newlfont}
\usepackage{amssymb}
\usepackage{amsmath}
\usepackage{latexsym}
\usepackage{amsthm}
\usepackage{eucal}
\usepackage{eufrak}
\usepackage{bbold}
 \usepackage{tocbibind}
 \usepackage{newlfont}
 \usepackage{enumerate}
 \usepackage{listings}
 \usepackage{tikz}
 \usepackage{bbm}
 \usepackage{hyperref}
 \usepackage{makecell}
 \usepackage{cleveref}
 \usepackage{float}
 \usepackage{enumitem}
\usepackage[caption = false]{subfig}

\usepackage[paper=a4paper,left=27mm,right=27mm,top=30mm,bottom=30mm]{geometry}

\usepackage[round]{natbib}

\usepackage{color}
\definecolor{mygreen}{RGB}{28,172,0} 
\definecolor{mylilas}{RGB}{170,55,241}

\theoremstyle{plain}                   
\newtheorem{theorem}{\bf Theorem}[section]
\newtheorem{lemma}[theorem]{\bf Lemma}

\newtheorem{definition}[theorem]{\bf Definition}
\newtheorem{remark}[theorem]{\bf Remark}
\newtheorem{proposition}[theorem]{ Proposition}

\newtheorem{assumption}[theorem]{\bf Assumption}

\DeclareMathOperator\supp{supp}

\newcommand{\bE}{\mathbb E}
\newcommand{\F}{{\cal F}}
\newcommand{\bF}{{\mathbb F}}

\newcommand{\M}{{\cal M}}

\newcommand{\bN}{\mathbb N}

\newcommand{\X}{{\cal X}}

\newcommand{\reals}{\mathbb{R}}

\newcommand{\comp}{\raise 1pt \hbox{$\scriptstyle\circ$}}

\newcommand{\upto}{{\raise 1pt \hbox{$\scriptstyle \,\nearrow\,$}}}
\newcommand{\downto}{{\raise 1pt \hbox{$\scriptstyle \,\searrow\,$}}}

\newcommand*\samethanks[1][\value{footnote}]{\footnotemark[#1]}

\newcommand{\black}[1]{\textcolor{black}{#1}}
\newcommand{\blue}[1]{\textcolor{black}{#1}}

 \title {Detecting asset price bubbles using deep learning}
 \author{Francesca Biagini\thanks{Workgroup Financial and Insurance Mathematics, Department of Mathematics, Ludwig-Maximilians Universit{\"a}t, Theresienstrasse 39, 80333 Munich, Germany. Emails: biagini@math.lmu.de,  meyer-brandis@math.lmu.de.} \and  Lukas Gonon\thanks{Department of Mathematics, Imperial College London, London, SW7 1NE UK. Email: l.gonon@imperial.ac.uk} \and Andrea Mazzon\thanks{Department of Economics, University of Verona, 37129 Verona, Italy. Email: andrea.mazzon@univr.it}  \and Thilo Meyer-Brandis \samethanks[1]}
                                       
\begin{document}

\maketitle

\begin{abstract}
In this paper we employ deep learning techniques to detect financial asset bubbles by using observed call option prices. The proposed algorithm is widely applicable and model-independent. We test the accuracy of our methodology in numerical experiments within a wide range of models and apply it to market data of tech stocks in order to assess if asset price bubbles are present. Under a given condition on the pricing of call options under asset price bubbles, we are able to provide a theoretical foundation of our approach for positive and continuous stochastic asset price processes. When such a condition is not satisfied, we focus on local volatility models. To this purpose, we give a new necessary and sufficient condition for a process with time-dependent local volatility function to be a strict local martingale.    
\end{abstract}

\section{Introduction}
The study of methodologies for detecting asset price bubbles has attracted an increasing interest over the last years in order to prevent, or mitigate, the financial distress often following their burst. \black{A financial bubble is defined as the temporary deviation of the market price of an asset from its fundamental value. }
In the context of the martingale theory of bubbles (see among others \cite{cox2005local}, \cite{LoewensteinWillard}, \cite{JarrowProtter2007},   \cite{JarrowProtter2010}, \cite{jarrow2011detect}, \cite{Biagini}, \cite{Protter2013}), the market price of a (discounted) asset represents a bubble if it is given by a strict local martingale. In this framework, detecting an asset price bubble amounts to determine if the stochastic process modelling the discounted asset price is a true martingale or a strict local martingale.

Our contribution is part of a wide range of works on asset price bubbles detection. The problem is tackled within the framework of local volatility models in the papers  \cite{jarrow2016testing}, \cite{jarrow2011detect} and \cite{jarrow2011there}, via volatility estimation techniques. In \cite{fusari2020testing}, under the joint assumption of \emph{no free lunch with vanishing risk} (NFLVR) and \emph{no-dominance}, the authors pursue this goal by looking at the differential pricing between put and call options, with the motivation that traded call and put options reflect price bubbles differently. On the other hand, a nonparametric estimator of asset price bubbles with the only assumption of NFLVR is proposed in \cite{jarrow2021inferring}, using a cross section of European option prices. The approach is based on a nonparametric identification of the state-price distribution and the fundamental value of the asset from option data. In \cite{piiroinen2018asset}, the authors introduce a statistical indicator of asset price bubbles based on the bid and ask market quotes for the prices of put and call options. This methodology is in particular applied to the detection of asset price bubbles for SABR dynamics. Finally, an asymptotic expansion of the right wing of the implied volatility smile is used in \cite{jacquier2018implied} to determine the strict local martingale property of the underlying.

In this paper, we propose \black{a new approach that formulates bubble detection as a supervised learning problem. The learning problem is solved using neural networks that take as inputs observed call option prices}. \black{In this way} we provide a theoretical foundation and a model-free deep learning approach \black{for bubble detection, which is new in the literature}. More precisely, we feed the network with a set of training data whose $i$-th data point consists of analytical call option prices for different strikes $K_1, \dots, K_n$ and maturities $T_1, \dots, T_m$ \black{written on} an underlying process $X^i$ with known dynamics, together with an indicator specifying if $X^i$ is a strict local martingale or a true martingale. %We then make the network assess if a new underlying (ideally coming from market data) has a bubble or not looking at the call option prices it produces for strikes $K_1, \dots, K_n$ and maturities $T_1, \dots, T_m$.
To assess if a new stock (possibly corresponding to market data) has a bubble or not, we evaluate the trained neural network at
(market) prices of call options on this underlying for strikes $K_1, \dots, K_n$ and maturities $T_1, \dots, T_m$.
The output \blue{of} the network is in \blue{$(0,1)$} and indicates the probability that the underlying exhibits a bubble.
Slightly changing the algorithm, we are also able to estimate the size of the bubble.

\black{In this manner the methodology does not attempt to detect a bubble based on the trajectory of the underlying, but instead uses the information contained in all call option prices available at a given point in time for different maturities and strikes to recover information on the asset's price probability distribution. In fact, the (theoretical) observation of call option prices for all strikes and maturities determines full information about the distribution of the underlying by the Breeden-Litzenberger formula. The method aims to exploit the information contained in call option prices in the best possible way to provide an estimate for the probability that an asset price bubble is present.}

\black{Formulating bubble detection as a supervised learning problem leads to a purely sample-based procedure. The only information used when training the neural network is a collection of call option prices associated to the process $X^i$ and an indicator stating if these call option prices correspond to a price process with a bubble or not. Thus, we can arbitrarily combine different models for bubbles and could also enhance synthetic model training data by historical data.}

One of the main advantages of our approach is that it does not require the direct estimation of any parameter or quantity related to the asset price process, and that it is model independent. In particular, we show via numerical experiments that our method works not only within a class of models, but also if the network is trained using a certain kind of stochastic processes (like local volatility models) and tested within another class (for example, stochastic volatility models). 
We also test the method with market data associated to assets \black{that have been suspected to be involved in the new tech bubble} (see among others \black{ \cite{fusari2020testing}, \cite{piiroinen2018asset}, \cite{kolakowski}, \cite{libich2021bitcoin}}, \cite{Bercovici}, \cite{Ozimek}, \cite{Seria}, \cite{Sharma}). According to the neural network\blue{,} output bubbles were present in these assets with high probability at certain points in time, and these dates match retrospectively expected results.

Our motivation for this method is manifold: on the one hand, the price of call options on a bubbly underlying has an additional term which is added to the usual risk neutral valuation due to a collateral requirement represented by a constant $\alpha \in [0,1]$, see \cite{cox2005local} and Theorem \ref{thm:coxresult} in our paper. Looking at call option prices, it is then theoretically possible to assess if the underlying has a bubble by identifying the presence of this term. On the other hand, in the case when the underlying price follows a local volatility model, a modified version of Dupire's formula stated in Theorem 2.3 of \cite{ekstrom2012dupire} permits to theoretically recover the local volatility function, which is crucial to determine if the process is a strict local martingale or a true martingale, from the observation of call option prices. 

%In particular, in Proposition \ref{prop:existenceF} we prove the existence of a function $F$ from the space of prices of call option prices written on local volatility models to the set $\{0,1\}$, whose value is $1$ if and only if the underlying has a bubble. 
In this way we are able to provide a theoretical foundation for our method. In particular, in Theorems \ref{theo:existencenetwork} and \ref{theo:existencenetworklocvol} we prove the existence of a sequence of neural networks approximating a theoretical ``bubble detection function'' $F$, {under the assumption $\alpha<1$ for underlyings given by continuous and positive stochastic processes and in the general case $\alpha \in [0,1]$ for local volatility models.} The function $F$ maps from the space $\X$ of call option prices as functions of  strike and maturity to $\{0,1\}$, being $1$ if and only if the underlying process is a strict local martingale. 
%This result is proved  for a general collateral requirement constant $\alpha \in [0,1]$ in the case when the underlying is represented by a local volatility model, and also holds for more general processes if $\alpha<1$. 
Several steps are required in order to prove {Theorems \ref{theo:existencenetwork} and \ref{theo:existencenetworklocvol}}. First, we show the existence of the function $F$ introduced above and show that it is measurable with respect to a natural topology on $\X$, see {Propositions \ref{prop:existenceF} and \ref{prop:existenceFlocal}}. In the case when $\alpha=1$, we consider the class of local volatility models under a fairly general assumption on the local volatility function. Specifically, we use some new findings providing a sufficient and necessary condition for a stochastic process with time dependent local volatility function to be a strict local martingale. The main result of this analysis is stated in Theorem \ref{thm:Xmartingaleifconditionforanyt} and is of independent interest, since it is a generalisation of  Theorem 8 of \cite{jarrow2011detect}. In {Propositions \ref{prop:existenceapproximation} and \ref{prop:existenceapproximationlocvol}} we then construct a sequence of measurable functions $(F^n)_{n \ge 1}$ approximating $F$ pointwise, where $F^n:\mathbb{R}^n \to [0,1]$ takes the prices of a call option with fixed underlying for $n$ different strikes and maturities and outputs the likelihood that the underlying has a bubble according to the observation of these prices. {These propositions are} then used together with the universal approximation theorem, see \cite{hornik1991}, to prove {Theorems \ref{theo:existencenetwork} and \ref{theo:existencenetworklocvol}, respectively}.

The paper is structured as follows. Section \ref{sec:theoretical} is devoted to the theoretical foundation of our approach. In particular, \black{in Section \ref{sec:bubbles} we briefly recall the martingale theory of asset price bubbles,} in Section \ref{sec:options} we introduce call option pricing in presence of financial asset bubbles, and in \ref{sec:neuralapprox} we  prove existence of the test function $F$ and of its neural network approximation when $\alpha<1$, when the asset price is given by a continuous and positive stochastic process.  On the other hand, in Section \ref{sec:localvolatilitymodels} we focus on local volatility models to cover the general case $\alpha \in [0,1]$: in Section \ref{sec:strictlocalvol} we provide sufficient and necessary conditions for the strict local martingale property of local volatility models, whereas in Section \ref{sec:neuralapproxlocal} we give results analogous to the ones in Section \ref{sec:neuralapprox} under this new setting.

In Section \ref{sec:numexperiments} we test the performance of our methodology via numerical experiments. Specifically, \black{we introduce our methodology in Section \ref{sec:nummethodology} and} in Section \ref{sec:localvol} we both feed and test the network within local volatility models, with different, randomly chosen parameters, whereas in Section  \ref{sec:stochvol} we feed the network with data coming from stochastic volatility models and test it with local volatility models, and vice-versa. \black{In Section \ref{sec:netarch} we comment on the choice of the network architecture}. In Section \ref{sec:marketdata} we apply our approach to the analysis of \black{four} stocks that have been claimed to be affected by a bubble in recent years, namely Nvidia, \black{Meta, GameStop and} Tesla. In Section \ref{sec:conclusions} we present some conclusions.

\section{Theoretical foundation}\label{sec:theoretical}
\subsection{Asset price bubbles}\label{sec:bubbles}
\black{In the literature there have been many approaches proposed for the mathematical modelization of asset price bubbles.}

 \black{From an economic point of view, the main challenge  consists in explaining how bubbles are generated at the  microeconomic level by the interaction of market participants; see  \cite{tirole},  \cite{hk},  \cite{DeLong},  \cite{Scheinkman},   \cite{Abreu}, \cite{foellmer} and the references therein.
From a mathematical point of view, asset price bubbles are mainly studied by using the \emph{local martingale framework}:  \cite{LoewensteinWillard},  \cite{cox2005local},  \cite{JarrowMadan}, \cite{JarrowProtter2007}, \cite{JarrowProtter2010},   \cite{JarrowKchiaProtter}, \cite{JarrowProtter2009}, \cite{JarrowProtter2011}, \cite{prrel}, \cite{PalProtter}, \cite{KardarasKreherNikeghbali},  \cite{Biagini}. 
 Another definition has been proposed in  \cite{herdegen2016strong}, where the fundamental value is assumed to be the superreplication price of the asset.
 In  \cite{JarrowProtter2012}, \cite{BiaginiNedelcu}  and  \cite{biagini_maz_mb}  constructive models are proposed to take in account triggering factors for bubble formation and bubble bursting. 
 For a survey, we refer to \cite{Protter2013}.}

Here  we follow the approach of \cite{JarrowProtter2009}, \cite{JarrowKchiaProtter}. The notion of an asset price bubble has two ingredients: the observed \emph{market price} $S$ of a given financial asset and the asset's \emph{fundamental value}. 
Assume discounting equal $1$. Given an equivalent martingale measure $Q$ for $S$, the fundamental value $S^Q$ is usually defined as the expected sum of future dividends under $Q$. 
 The bubble perceived under $Q$ is defined as the difference between the two:
$$\beta^Q= S -S^Q.$$
In this setting, we have a bubble if and only if $S$ is a strict local martingale under $Q$. In the sequel, we present  a new method to detect financial asset prices' bubbles by testing for strict \blue{local} martingale properties based on observations of call option prices and 
by applying machine learning techniques.

\subsection{European call options under asset price bubbles}\label{sec:options}

Let $(\Omega, P, \F)$ be a probability space endowed with a filtration $\bF = (\F_t)_{t \ge 0}$ satisfying the usual hypothesis. It is well known that the risk-neutral valuation of European options when the underlying asset price has a bubble (i.e., when it is a strict local martingale) does not satisfy the put-call parity, see for example \cite{JarrowProtter2010}. An alternative way to price call options for a bubbly underlying, also supported by market data, has been proposed in \cite{cox2005local}. 

We start with the following definition.

\begin{definition}\label{def:ourcallprice} 
Introduce an underlying given by a continuous, positive stochastic process $X=(X)_{t \ge 0}$, and fix a pricing measure $Q \sim P$ under which $X$ is a local martingale. The price $C^{\alpha, X}(T,K)$ at time $t=0$ of a call option of maturity $T$ and strike $K$  written on $X$ with collateral constant $\alpha \in [0, 1]$ is given by
\begin{equation}\label{eq:pricecallcollateralized}
C^{\alpha, X}(T,K) = \bE^Q[(X_T-K)^+]+ \alpha m(T),
\end{equation}
where 
\begin{equation}\label{eq:martingaledefect}
m(T) := X_0 - \bE^Q[X_T].
\end{equation}
 \end{definition}
   The term $m(T)$ in \eqref{eq:martingaledefect} is called the \emph{martingale defect} of $X$. Note that $m(T)>0$ if $X$ is a strict local martingale bounded from below and $m(T)=0$ if $X$ is a martingale.
   
\begin{remark}
Note that if for $t\in [0,T]$ we set
$$C^{\alpha,X}_t(T,K):= \bE^Q[(X_T-K)^{+}| \mathcal{F}_t] +\alpha m_t(T)$$
for $m_t(T)=X_t - \bE^Q[X_T| \mathcal{F}_t]$, then $C^{\alpha,X}_t(T,K)$, $t\in [0,T]$, is a non-negative local martingale, which equals $(X_T -  K)^{+}$ at final time.
Hence $C^{\alpha,X}_t(T,K)$ represents an arbitrage-free price at time $t$ for the call option $(X_T-K)^{+}$.
\end{remark}

%{Definition \ref{def:ourcallprice} is also valid for incomplete markets. In this case, $Q$ can be interpreted as the pricing measure \emph{chosen by the market} as in \cite{JarrowProtter2010}. 
If the market is complete, a financial interpretation of \eqref{eq:pricecallcollateralized}-\eqref{eq:martingaledefect} where $\alpha \in [0,1]$ assumes the meaning of \emph{collateral requirement} comes from \cite{cox2005local}, where the following definition is given.

\begin{definition}\label{def:cox}
 The price at time $t=0$ of a call option $C=(X_T-K)^+$ of maturity $T$ and strike $K$ written on a continuous, positive underlying $X=(X_t)_{t \ge 0}$ is the smallest initial value $V_0$ of a superreplicating portfolio $V =(V_t)_{t \ge 0}$ satisfying the admissibility condition 
 \begin{equation}\label{eq:coxpric}
 V_t \ge \alpha(X_t-K)^+
 \end{equation}
for all $t \in [0,T],$  where $\alpha \in [0, 1]$ is a constant representing the \emph{collateral requirement}. 
  \end{definition}
  Condition \eqref{eq:coxpric} is always satisfied when the underlying asset has no bubble. For {complete market models} we have the following result, given in \cite{cox2005local}.

 \begin{theorem}\label{thm:coxresult} 
{Introduce an underlying given by a continuous, positive stochastic process $X=(X)_{t \ge 0}$, and assume the existence of a unique probability measure $Q \sim P$ under which $X$ is a local martingale}. The price {according to Definition \ref{def:cox}} at time $t=0$ of a call option of maturity $T$ and strike $K$  written on $X$ and with collateral requirement $\alpha \in [0, 1]$ is given by equation \eqref{eq:pricecallcollateralized}.
 \end{theorem}
Such results can be extended to c\`adl\`ag processes, see  \cite{cox2005local}. 
 \begin{remark}
Under a condition known as \emph{no dominance}, the price of a call option on a bubbly asset is given by formula \eqref{eq:pricecallcollateralized} with $\alpha=1$ and satisfies the put-call parity, see Theorem 6.2 of \cite{JarrowProtter2010}. Our setting, where $\alpha \in [0, 1]$, includes the framework of \cite{JarrowProtter2010} as a particular case. {Moreover, our method does not need to know which value of $\alpha$ is taken to price the option.}
 \end{remark}

\subsection{Neural network approximation for collateral $\alpha \in [0,1)$}\label{sec:neuralapprox}
In this section we fix a collateral constant $\alpha \in [0,1)$ and suppose that the price of a call option is given by \eqref{eq:pricecallcollateralized}, with continuous and positive underlying $X$. For the sake of simplicity, from now on we write $C^{X}$ for the call option price. We treat the case $\alpha=1$ in Section \ref{sec:localvolatilitymodels} in the specific setting of local volatility models.  Our aim is to find a way of detecting asset price bubbles by using neural networks, as we explain in the sequel. {From now on, we fix the ``pricing measure" $Q \sim P$}.

%Our aim is to prove the existence of a sequence of neural network maps $(\hat{F}^n)_{n \in \bN}$, $\hat{F}^n \colon (0,\infty)^n \to [0,1]$, see Definition \ref{def:neuralnetwork}, with the following property: if $x$ is a vector of call option prices written on an underlying $X$ satisfying \eqref{eq:localvolatility}, with collateral constant $\alpha$ and maturities and strikes $(T_i,K_i)$, $i=1,\dots,n$, chosen well enough, then $\hat{F}^n(x) \approx 1$ if $X$ is a strict local martingale under the pricing measure $Q$ and $\hat{F}^n(x) \approx 0$ otherwise, with accuracy increasing with $n$ and tending to $100\%$ for $n \to \infty$. Once such a sequence is given, we can use $\hat{F}^n$ to detect bubbles: if we are given observed market prices of call options $\hat{C}(T_i,K_i)$, $i=1,\dots, n$, we evaluate  $\hat{F}^n(\hat{C}(T_1,K_1),\ldots,\hat{C}(T_n,K_n))$ to detect if there is a bubble or not. 

%From now on, we assume without loss of generality that the option prices are given by equation \eqref{eq:pricecallcollateralized} with $\alpha = 1$.  
%
 We start our analysis with the following result. 

\begin{proposition}\label{prop:existenceF}
{Fix a constant $x_0>0$}. Consider the spaces
\begin{align*}
  \M_{loc}=\bigl\{
       X=(X_t)_{t \ge 0} \text{ }
   & 
   \text{continuous and positive local martingale on $(\Omega, \F, Q, \bF)$ {with $X_0=x_0$}}
   %&    \text{{such that the function $t \to \bE^Q[(X_t-K)^+]$ is continuous for any $K \ge 0$}}
    \bigr\}
\end{align*}
and 
{\begin{equation}\label{eq:localmartingalesspace}
\M_{L}= \{\text{$X \in \M_{loc}$: $X$ is a strict local martingale}\}.
\end{equation}}
Introduce the set
\begin{equation}\label{eq:mathcalX}
\X := \left\{C^X: \text{ } X \in \M_{loc} \right\} {\cap C([0,\infty) \times (0,\infty),[0,\infty)})
\end{equation}
endowed with the Borel sigma-algebra of ${C([0,\infty) \times (0,\infty),[0,\infty))}$. Then there exists a measurable function $F \colon \X \to \{0,1\}$ such that $F(C^X) = \mathbbm{1}_{\{X \in \M_L\}}$. 

In particular,  
 \begin{equation}\label{eq:Ftodectedbubble1}
 F(f)\black{=} \mathbbm{1}_{\{\text{there exists $t \ge 0$ such that $\lim_{x \to 0}f(T,x)<X_0$ for any $T \ge t$}\}}, \quad f \in \X.
 \end{equation}
\end{proposition}

\begin{proof}
Let $X \in  \M_{loc}$. From \eqref{eq:pricecallcollateralized} and \eqref{eq:martingaledefect} we get
\begin{align}\label{eq:limitto0}
\lim_{K \to 0}C^{X}(T,K) &= \lim_{K \to 0} {\bE^Q[(X_T-K)^+]} + \alpha m(T) =  \bE^Q[X_T]+ \alpha(X_0 - \bE^Q[X_T]) \notag \\
&=  \alpha X_0 + \bE^Q[X_T](1-\alpha),
\end{align}
by the Dominated Convergence Theorem since $X$ is a supermartingale and therefore integrable {for any $t \geq 0$}.
In particular, since $X$ is a supermartingale, equation \eqref{eq:limitto0} implies that it is a strict local martingale if and only if there exists {$t > 0$} such that 
$$
\lim_{K \to 0}C^{X}(T,K)<X_0 \text{ for any {$T \ge t$}}.
$$
 Therefore,  we have 
 $$
 F(C^X) = {\mathbbm{1}_{\{X \in \M_{L}\}}} \text{ for any $C^{X} \in \X$}
 $$
 for $F$ given in \eqref{eq:Ftodectedbubble1}. In order to prove that $F$ is measurable, note that
  $$
 F(f)= 1- \prod_{t \in \bN}\bar F_t(f), \quad f \in \X,
 $$
 with $\bar F_t : \X \to \{0,1\}$  defined by 
 $$
 \bar F_t(f):=\mathbbm{1}_{\{\lim_{x \to 0}f(t,x) \ge X_0\}}=\mathbbm{1}_{\{\lim_{x \to 0}\bar{G}_{t,x}(f) \ge 0\}}, \quad t \in \bN, \quad f \in \X,
 $$ 
 where $\bar{G}_{t,x} : \X \to \reals$ is given by $\bar{G}_{t,x}(f):=f(t,x)-X_0$. As $\bar{G}$ is measurable, then also $\lim_{x \to 0}\bar{G}_{t,x}$ is measurable, and hence $\bar F_t$ and consequently $F$ are also measurable.
\end{proof}
We now prove the existence of a sequence of functions $(F^n)_{n \in \bN}$, $F^n \colon \reals^n \to [0,1]$, such that $F^n$ approximates, in some way we specify, the function $F$ from \eqref{eq:Ftodectedbubble1}.

 \begin{proposition}\label{prop:existenceapproximation}
 Define the set
%\begin{equation}\label{eq:setofmaturitiesandstrikes}
%A = \begin{cases}
%[0,\infty) \times (0,\epsilon] & \quad \text{if $\alpha<1$}, \\
%[0,\infty) \times [C,\infty) & \quad \text{if $\alpha=1$}.
%\end{cases}
%\end{equation}
\begin{equation}\label{eq:setofmaturitiesandstrikes}
A :=[0,\infty) \times (0,\epsilon]
\end{equation}
for some fixed $\epsilon>0$.

Let $(T_n, K_n)_{n \in \bN}$  be a sequence of maturities and strikes such that the set $ \{(T_n, K_n), n \in \bN\}$ is dense in $A$. Let $\mu$ be a probability measure on $\X$  such that there exists a probability measure $\nu$ on $\X$ with
	\begin{equation}\label{eq:measurenotzero}
	\nu\left(\left\{f\in \X \text{ such that } f(T_i, K_i) = g(T_i,K_i)\text{ for all $i \in \bN$}\right\}\right)>0 \quad \text{for any $g \in \supp(\mu)$}.
	\end{equation}
Then there exists a sequence of functions $(F^n)_{n \in \bN}$, $F^n \colon \reals^n \to [0,1]$ such that%which converges to $F$ pointwise, i.e., such that for any fixed $\sigma \in \Sigma$ it holds
\begin{equation}\label{eq:limitproposition}
\int_{\X} \left| F^n(g(T_1,K_1),\ldots,g(T_n,K_n))-F(g) \right|^pd\mu(g)  \xrightarrow[n \to \infty]{}0,
\end{equation}
for any $p \in [1,\infty)$, where $F \colon \X \to \{0,1\}$  is the function introduced in \eqref{eq:Ftodectedbubble1}.
%In particular, given a probability measure $\mu$ on $(\Sigma,\B^{\Sigma})$, for any $\varepsilon >0$ there exists $n_0 \in \bN$ such that for all $n \geq n_0$ we have 
%\begin{equation}\label{eq:resultconvergence}
%\int_{\Sigma}|F^n(C^\sigma(T_1,K_1),\ldots,C^\sigma(T_n,K_n))  - F(C^\sigma)| \mu(d \sigma) <\varepsilon.
%\end{equation}
For fixed $n \in \mathbb{N}$, the function $F^n$ can be chosen as follows:
%	\begin{equation}\label{eq:Fn}
%	F^n(c):= \frac{1}{|S_n(c)|} \sum_{\tilde \sigma \in S_n(c)} F\left(C^{\tilde{\sigma}}\right)
%	\end{equation}
%	where
%	\begin{equation}\label{eq:Snc}
%	S_n(c) := \left\{ \tilde \sigma \in \Sigma \text{ such that } C^{\tilde{\sigma}}(T_i, K_i) = c_i \text{ for all $i = 1, \dots, n$}\right\}.
%	\end{equation}
\begin{align}\label{eq:Fn}
	F^n(c^n):= \begin{cases}
	\frac{1}{\nu\left(S_n(c^n)\right)} \int_{S_n(c^n)}F(f)d\nu(f) \quad &\text{if $\nu\left(S_n(c^n)\right) > 0$},\\
	0 \quad &\text{otherwise},
	\end{cases}
	\end{align}
		$c^n \in \reals^n$, where
	\begin{equation}\label{eq:Snc}
	S_n(c^n) := \left\{f\in \X \text{ such that } f(T_i, K_i) = c^n_i \text{ for all $i = 1, \dots, n$}\right\}, \quad c^n \in \reals^n.
	\end{equation}
	\end{proposition}
\begin{proof}
 Let $(T_n, K_n)_{n \in \bN}$  be a sequence of maturities and strikes such that the set $ \{(T_n, K_n), n \in \bN\}$ is dense in $A$ given in \eqref{eq:setofmaturitiesandstrikes}. 
 Let $\mu$ be the probability measure on $\X$ introduced above and fix $g \in \supp(\mu)$. Also let $\nu$ be a probability measure on $\X$ which satisfies \eqref{eq:measurenotzero}. For any $n \in \bN$, consider the function $F^n \colon \reals^n \to [0,1]$ defined in 
 \eqref{eq:Fn} and set 
 $$
 c^{g,n} :=\left(g(T_1,K_1), \dots, g(T_n,K_n)\right).
 $$ 
We first want to show that 
 \begin{equation}\label{eq:FntoFinfinity}
 F^n(c^{g,n}) \xrightarrow[n \to \infty]{} F^{\infty}(c^{g, \infty}),
 \end{equation}
  with
  $$
c^{g, \infty}:=\{ g(T_i,K_i), \quad i \in \bN\}
$$
 and
 $$
F^{\infty}(c^{g, \infty}) :=  \frac{1}{\nu(S_{\infty}(c^{g,\infty}))} \int_{S_{\infty}(c^{g, \infty})}F(f)d\nu(f),
$$
where 
\begin{equation}\notag
	S_{\infty}(c^{g, \infty}) := \left\{f\in \X \text{ such that } f(T_i, K_i) = g(T_i,K_i)\text{ for all $i \in \bN$}\right\}
	\end{equation}
 satisfies $\nu(S_{\infty}(c^{g,\infty})) \ne 0$ by \eqref{eq:measurenotzero}.
Since $S_{n+1}(c^{g,n+1}) \subset S_n(c^{g,n})$ for any $n \in \mathbb{N}$ and $S_{\infty}(c^{g, \infty})= \bigcap\limits_{n \in \bN} S_n(c^{g,n})$, we have
\begin{equation}\label{eq:convergenceindicators}
\mathbb{1}_{S_n(c^{g,n})}  \xrightarrow[n \to \infty]{}\mathbb{1}_{S_{\infty}(c^{g,\infty})} \quad \text{pointwise}
\end{equation}
and
\begin{equation}\label{eq:convergencemeasures}
\nu\left(S_n(c^{g,n})\right)\xrightarrow[n \to \infty]{}\nu\left(S_{\infty}(c^{g, \infty})\right).
\end{equation}
By \eqref{eq:convergenceindicators} and since $|F| \le 1$, we can apply the Dominated Convergence Theorem and get
\begin{equation}\label{eq:convergenceintegrals}
\int_{S_n(c^{g,n})}F(f)d\nu(f) \xrightarrow[n \to \infty]{}\int_{S_{\infty}(c^{g, \infty})}F(f)d\nu(f).
\end{equation}
Putting together \eqref{eq:convergencemeasures} and \eqref{eq:convergenceintegrals}, we obtain \eqref{eq:FntoFinfinity}.

%Suppose first $\alpha=1$ and let $C>0$ be the large constant of the statement of the proposition. 
Moreover, as the set $ \{(T_n, K_n), n \in \bN\}$ is dense in $A$ by assumption, $g$ is continuous and every function in $\X$ is continuous as well, we have that
$$
S_{\infty}(c^{g, \infty})= \left\{f\in \X \text{ such that } f(T, K) = g(T,K) \text{ for all $(T,K) \in A$}\right\}.
$$
By the definition of $F$ in \eqref{eq:Ftodectedbubble1}, this implies that $F(f)=F(g)$ for any $f \in S_{\infty}(c^{g, \infty})$, and so $F^{\infty}(c^{g, \infty})=F(g)$.
By \eqref{eq:FntoFinfinity}, it thus follows that
$$
F^n(g(T_1,K_1),\ldots,g(T_n,K_n)) \xrightarrow[n \to \infty]{} F(g).
$$
We then obtain \eqref{eq:limitproposition} by the Dominated Convergence Theorem, since $(F^n)_{n \in \mathbb{N}}$  is bounded by $1$.
\end{proof}
 \begin{remark}
 Any discrete probability measure $\mu$ on $\X$, i.e. any probability measure $\mu$ such that $\supp(\mu) \subseteq \X_0$ for a countable set $\X_0 \subset \X$, satisfies the requirement of Proposition \ref{prop:existenceapproximation} that there exists a probability measure $\nu$ on $\X$ for which \eqref{eq:measurenotzero} holds. An example is given by
 $$
 \nu := \frac{1}{ \sum_{n \in \mathbb{N}} n^{-2}} \sum_{n \in \mathbb{N}} \delta_{g_n}n^{-2}, 
 $$
  where $\X_0:=\{g_n, n \in \mathbb{N}\}$. Indeed, for any $g \in \supp(\mu)$ there exists $n \in \mathbb{N}$ such that $g=g_n$, thus  $$
  \nu\left(\left\{f\in \X \text{ such that } f(T_i, K_i) = g(T_i,K_i)\text{ for all $i \in \bN$}\right\}\right) \ge   \nu\left(\left\{g_n\right\}\right) >0.
  $$ 
 %If the probability measure $\mu$ introduced in Proposition \ref{prop:existenceapproximation} is discrete, i.e., $\supp(\mu) \subseteq \Sigma_0$ for a countable set $\Sigma_0 \subset \Sigma$, we can construct  
 \end{remark}
 
{We now want to approximate the sequence of functions} $F^n : \reals^n \to \black{[0,1]}$ defined in \eqref{eq:Fn}, {and thus the function $F$, by a sequence of neural networks}. We start with the formal definition of a neural network, see for example  \cite{hornik1991}.
\begin{definition}\label{def:neuralnetwork}
Let $I \subseteq \reals$. A neural network $\hat{F} \colon \reals^{n} \to I$ is a function
\[ \hat{F}(x)=(\psi \circ W_L \circ ({\sigma} \circ W_{L-1}) \circ \cdots \circ ({\sigma} \circ W_1))(x), \quad x \in \reals^{n}, \]
with given $L \in \bN$, $L \ge 2$, where the inner activation function $\sigma \colon \reals \to \reals$ is bounded, continuous and non-constant, where the outer activation function $\psi \colon \reals \to I$ is invertible and Lipschitz continuous and where $W_\ell \colon \reals^{N_{\ell-1}} \to \reals^{N_\ell}$ are affine functions with $N_1,\ldots,N_{L-1} \in \bN$, $N_0=n$, $N_L=1$.
\end{definition} 

We then have the following result.

\begin{theorem}\label{theo:existencenetwork}
 Fix $p\in [1,\infty)$  and let $(T_n, K_n)_{n \in \bN}$  be the sequence of maturities and strikes introduced in Proposition \ref{prop:existenceapproximation}. Also let $\mu$ be a probability measure on $\X$ such that there exists a probability measure $\nu$ on $\X$ which satisfies \eqref{eq:measurenotzero}. Then for any $\epsilon >0$ there exists an $n \in \bN$ and a neural network $\hat{F}^n \colon \reals^n \to [0,1]$ such that
$$
\left(\int_\X |\hat{F}^n(g(T_1,K_1),...,g(T_n,K_n)) - F(g)|^p d\mu(g) \right)^{1/p}< \epsilon.
$$
% Fix $n \in \bN$, $p>0$ and let $F$ be the function defined in \eqref{eq:Ftodectedbubble1} and \eqref{eq:Ftodectedbubble2}. Then for any $\epsilon>0$ there exists a neural network $\hat{F}^n \colon \reals^{n} \to \{0,1\}$ such that
% $\lVert F - \hat F^n \rVert_p \le \epsilon$, where $\lVert \cdot \rVert_p$ is the $L^p$-norm.
\end{theorem}	 
\begin{proof} 
Fix $p\in [1,\infty)$ and $\epsilon>0$. By Proposition \ref{prop:existenceapproximation} there exists $n \in \bN$ such that
\begin{equation}\label{eq:firstepsilonhalf}
\int_\X |F^n(g(T_1,K_1),...,g(T_n,K_n)) - F(g)|^p d\mu(g) < \left(\frac{\epsilon}{2}\right)^p,
\end{equation}
with $F^n$ defined in \eqref{eq:Fn}. Let $\psi \colon \reals \to [0,1], $ invertible and Lipschitz continuous with Lipschitz constant $K>0$. By  the universal approximation theorem, see \cite{hornik1991}, there exists a neural network $\phi^n \colon \reals^{n} \to \reals$  with identity as outer activation function and such that
$$
\int_\X |\phi^n(g(T_1,K_1),...,g(T_n,K_n)) - (\psi^{-1} \circ \black{F^n}) (g(T_1,K_1),...,g(T_n,K_n)) |^p d\mu(g) < \left(\frac{\epsilon}{2K}\right)^p.
$$
Set now $\hat F^n := \psi\circ\phi^n$. Then we have
\begin{align}
&\int_\X |\hat F^n(g(T_1,K_1),...,g(T_n,K_n)) -\black{F^n} (g(T_1,K_1),...,g(T_n,K_n)) |^p d\mu(g) \notag \\ 
& \le K^p \int_\X |\phi^n(g(T_1,K_1),...,g(T_n,K_n)) - (\psi^{-1} \circ \black{F^n}) (g(T_1,K_1),...,g(T_n,K_n)) |^p d\mu(g) < \left(\frac{\epsilon}{2}\right)^p.\label{eq:secondepsilonhalf}
\end{align}
Therefore it holds
\begin{align}
&\left(\int_\X |\hat{F}^n(g(T_1,K_1),...,g(T_n,K_n)) - F(g)|^p d\mu(g)\right)^{1/p} \notag \\
& \le \left(\int_\X |\hat F^n(g(T_1,K_1),...,g(T_n,K_n)) -\black{F^n} (g(T_1,K_1),...,g(T_n,K_n)) |^p d\mu(g)\right)^{1/p}   \notag \\
& \quad + \left(\int_\X |F^n(g(T_1,K_1),...,g(T_n,K_n)) - F(g)|^p d\mu(g)\right)^{1/p}   \notag \\
& < \epsilon \notag
\end{align}
by \eqref{eq:firstepsilonhalf} and \eqref{eq:secondepsilonhalf}.
\end{proof}

\begin{remark}
When $\alpha=1$, we have
\begin{equation}\notag
\lim_{K \to 0}C^{X}(T,K) = \lim_{K \to 0} \bE^Q[(X_T-K)^+] + m(T) =   X_0,
\end{equation}
so the computations illustrated in the proof of Proposition \ref{prop:existenceF} do not help detecting the strict local martingale property of $X$. In this case, we prove a neural network approximation focusing on local volatility models.
\end{remark}

\begin{remark}
The function $F^n$ is defined on $\mathbb{R}^n$ with $n$ possibly very large. In addition to their universal approximation properties (see \cite{hornik1991}) neural networks have been shown to be particularly suitable for learning high-dimensional functions due to their ability to overcome the curse of dimensionality in many situations, see for example the seminal work \cite{Barron1993} and \cite{EGJS18_787}, \cite{HornungJentzen2018}, \cite{ReisingerZhang2019}, \cite{GS20_925}\blue{,} \cite{Gonon2021}  in the context of option pricing. These properties are a key motivation for our choice of neural networks to approximate $F^n$.
\end{remark}

 \subsection{The general case $\alpha \in [0,1]$ in the setting of local volatility models}\label{sec:localvolatilitymodels}

We now include the case $\alpha=1$ in equation \eqref{eq:pricecallcollateralized}. In this setting, we focus on the class of local volatility models, i.e., stochastic processes $X=(X_t)_{t \ge 0}$ such that there exists a risk neutral measure $Q \sim P$ under which $X$ has dynamics
\begin{equation}\label{eq:localvolatility}
dX_t = \sigma (t, X_t)dW_t, \quad t \ge 0, \quad X_0 = x >0,
\end{equation}
where $W=(W_t)_{t \ge 0}$ is a one-dimensional $(\bF,Q)$-Brownian motion and the function $\sigma$ satisfies the following conditions. {In the sequel, \emph{martingale} and \emph{(strict) local martingale} are always meant with respect to the probability measure $Q$. }
\begin{assumption}\label{ass:sigma}
The function $\sigma: [0,\infty) \times (0, \infty) \to (0,\infty)$ is continuous. Moreover, it is locally H\"older continuous with exponent $1/2$ with respect to the second variable, i.e., for any $L>0$ and any compact set $ [0,L] \times [L^{-1}, L]$ there exists a constant $C \in [0, \infty)$ such that 
$$
|\sigma(t,x)-\sigma(t,y)| \le C|x-y|^{1/2} 
$$
for all $(t,x,y) \in [0,L] \times [L^{-1},L]^2$. Furthermore, there exists a constant $A  \in [0, \infty)$ such that $\sigma(t,x) \le A$ for any $(t,x) \in [0,\infty) \times (0,1]$.  
\end{assumption}
Assumption \ref{ass:sigma} provides existence and uniqueness of a strong solution of \eqref{eq:localvolatility}, see Chapter IX.3 in \cite{revuz2013continuous}. %, which in particular admits a strictly positive and smooth density, see the proof of Theorem 2.2 in \cite{ekstrom2012dupire}.
In particular, we do not require the volatility function to be at most of linear growth in the spatial variable. As stated in the next result, which is Theorem 8 of \cite{jarrow2011detect}, this allows $X$ to be a strict local martingale. 
\begin{theorem}\label{thm:sigmatomartingales}
Let $\sigma: [0,\infty) \times (0, \infty) \to (0,\infty)$ be a function satisfying Assumption \ref{ass:sigma}. Assume that there exist two functions $\underline{\sigma}, \overline{\sigma} : (0,\infty) \to (0, \infty)$, which are continuous and locally  H\"older continuous with exponent $1/2$, such that $\underline{\sigma}(x) \le \sigma(t,x) \le \overline{\sigma}(x)$ for any $(t,x) \in [0,\infty) \times (0, \infty)$. Let $X=(X_t)_{t \ge 0}$ be the process in \eqref{eq:localvolatility}. The following holds:
\begin{enumerate}
\item If $\int_{c}^{\infty} \frac{y}{\overline{\sigma}^2(y)}dy = \infty$ for every $c>0$, then $X$ is a martingale.
\item If there exists $c>0$ such that  $\int_{c}^{\infty} \frac{y}{\underline{\sigma}^2(y)}dy < \infty$, then $X$ is a strict local martingale. 
\end{enumerate}
\end{theorem}
We now give a generalization of Theorem \ref{thm:sigmatomartingales} that is useful to provide the theoretical foundation of our methodology in this setting. 

 \subsubsection{Strict local martingale property of local volatility models}\label{sec:strictlocalvol}
Before stating the main result, we give two lemmas.
\begin{lemma}\label{lem:sigmatolocalmartingale}
Fix $c>0$. Let $\sigma: [0,\infty) \times (0, \infty) \to (0,\infty)$ be a function satisfying Assumption \ref{ass:sigma} and $X=(X_t)_{t \ge 0}$ be the process in \eqref{eq:localvolatility}. Then $X$ is a strict local martingale if there exist $s,T$ with $0\le s<T$ and a function $\underline{\sigma}: (0,\infty) \to (0,\infty)$, continuous and locally H\"older continuous with exponent $1/2$, such that:
\begin{align}
&\int_c^{\infty}  \frac{y}{\underline{\sigma}^2(y)}dy < \infty, \label{assum:integralfinite} \\
& \underline{\sigma}(x) \le \sigma(t,x) \text{ for any $(t,x) \in [s,T] \times (0,\infty)$}. \label{assum:smallerthansigma}
\end{align}
\end{lemma}
\begin{proof}
 Let $0 \le s<T$ and $\underline{\sigma}$ as above. For $x \in (0, \infty)$, define the process $X^{\underline{\sigma},T, s,x}$ to be the unique strong solution of 
\begin{align}
dX^{\underline{\sigma},T, s,x}_t &= \underline\sigma (X^{\underline{\sigma},T, s,x}_t)dW_t, \quad  s \le t \le T,\notag \\
 X^{\underline{\sigma},T, s,x}_s &= x.\notag
\end{align}
Lemma  9 of \cite{jarrow2011detect}  together with the assumption in  \eqref{assum:smallerthansigma} implies that
$$
{\bE^Q}[X_T|X_s=x] \le {\bE^Q}[X^{\underline{\sigma},T, s,x}_T] < x,
$$
where the last inequality comes from the fact that the process $X^{\underline{\sigma},T, s,x}$ is a strict local martingale for any $(s,x) \in [0,\infty) \times (0,\infty)$ by Theorem 7 of \cite{jarrow2011detect} and \eqref{assum:integralfinite}. 
\end{proof}

\begin{lemma}\label{lemma:sigmatomartingalesgen}
Fix $c>0$ and $(r,x) \in [0,\infty) \times (0,\infty).$ Let $\sigma: [0,\infty) \times (0, \infty) \to (0,\infty)$ be a function satisfying Assumption \ref{ass:sigma}. Denote by $X^{r,x}=(X^{r,x}_t)_{t \ge r}$  the unique strong solution of the SDE
\begin{align}
dX^{r,x}_t &= \sigma (t, X^{r,x}_t)dW_t, \quad  t \ge r, \notag \\
 X^{r,x}_r &= x.\notag
\end{align}
Then $X^{r,x}$ is a true martingale if for all $T > r$ there exists a function $\overline{\sigma}^{(T)}: (0,\infty) \to (0,\infty)$, continuous and locally H\"older continuous with exponent $1/2$, such that:
\begin{align}
&\int_c^{\infty}  \frac{x}{(\overline{\sigma}^{(T)}(x))^2}dx = \infty, \label{assum:integralinfinity} \\
& \overline{\sigma}^{(T)}(x) \ge \sigma(t,x) \text{ for any $(t,x) \in [r,T] \times (0,\infty)$}. \label{assum:biggerthansigma}
\end{align}
\end{lemma}
\begin{proof}
Consider $s,T$ such that $T>s \ge r$, and $y \in (0, \infty)$. Define the process $X^{\overline{\sigma},T, s,y}$ to be the unique strong solution of 
\begin{align}\notag
dX^{\overline{\sigma},T, s,y}_t &= \overline\sigma^{(T)} (X^{\overline{\sigma},T, s,y}_t)dW_t, \quad  s \le t \le T,\notag \\
 X^{\overline{\sigma},T, s,y}_s &= y.\notag
\end{align}
Lemma  9 of \cite{jarrow2011detect}  together with the assumption in  \eqref{assum:biggerthansigma} implies that
$$
{\bE^Q}[X^{r,x}_T|X^{r,x}_s=y] \ge {\bE^Q}[X^{\overline{\sigma},T, s,y}_T] = y,
$$
where the last equality holds because the process $X^{\overline{\sigma},T, s,x}$ is a martingale by Theorem 7 of \cite{jarrow2011detect} and \eqref{assum:integralinfinity}. Since $X^{r,x}$ is a local martingale and then a supermartingale, we obtain ${\bE^Q}[X^{r,x}_T|X^{r,x}_s=y] = y$. Therefore, $X^{r,x}$ is a martingale because $s, T$ and $y$ have been chosen arbitrarily. 
% now move to prove necessity. In particular, we suppose that there is $T>0$ such that there exists no continuous and H\"older continuous function with exponent $1/2$, such that both \eqref{assum:integralinfinity} and \eqref{assum:biggerthansigma} hold. We prove that $X$ is a strict local martingale under this assumption. Fix $s \in (0,T)$. Since  
\end{proof}

In Theorem \ref{thm:Xmartingaleifconditionforanyt} we now provide a sufficient and necessary condition for $X$ to be a true martingale in terms of the local volatility function $\sigma$. To this purpose we need requirements on $\sigma$ which are more restrictive than in Assumption \ref{ass:sigma}. However they are satisfied in the most commonly employed models in applications. 

\begin{assumption}\label{ass:monotone}
The function $\sigma: [0,\infty) \times (0, \infty) \to (0,\infty)$ satisfies Assumption \ref{ass:sigma} and there exist time points $(T_n)_{n \ge 0}$, with $T_0=0$ and $T_{n+1}>T_n$ for any $n\ge 0$, such that $\sigma(\cdot,x)$ restricted to any interval $[T_n,T_{n+1}]$, $n \ge 0$, is monotone increasing for all $x \in (0,\infty)$ or monotone decreasing for all $x \in (0,\infty)$.
\end{assumption}

Note that Assumption \ref{ass:monotone} includes the case when $\sigma$ is independent of time or is monotone increasing for all $x \in (0,\infty)$ or monotone decreasing for all $x \in (0,\infty)$ with respect to the time variable.

\begin{theorem}\label{thm:Xmartingaleifconditionforanyt}
Fix $c>0$, and let $\sigma: [0,\infty) \times (0, \infty) \to (0,\infty)$ satisfy Assumption \ref{ass:monotone}. Let $X=(X_t)_{t \ge 0}$ be the process in \eqref{eq:localvolatility}. Then $X$ is a true martingale if and only if 
\begin{align}
&\int_c^{\infty}  \frac{x}{\sigma^2(T,x)}dx = \infty \text{ for any $T \in (0,\infty) \setminus B^{\sigma}$,} \label{assum:integralinfinityforfixedT}
\end{align}
\end{theorem}
with 
\begin{equation}\label{eq:Bsigma}
B^{\sigma}:=\left\{T>0: \text{$\sigma(\cdot,x)$ has a local maximum at $T$ for all $x \in (0,\infty)$}\right\}.
\end{equation}
\begin{proof}
We start proving that if $X$ is a martingale, then \eqref{assum:integralinfinityforfixedT} holds. Suppose there exists $T \in (0, \infty) \setminus B^{\sigma}$ such that 
\begin{align}
&\int_c^{\infty}  \frac{x}{\sigma^2(T,x)}dx < \infty.\notag
\end{align}
Assume first that $T \in (T_n, T_{n+1})$ for some $n \in \bN \cup \{0\}$. Then $\underline{\sigma}(x):=\sigma(T,x)$ is a lower bound for $\sigma$ on the interval $[T, T+\delta]$ (on the interval $[T-\delta, T]$) for some $\delta>0$ if $\sigma(\cdot,x)_{|_{[T_{n},T_{n+1}]}}$ is monotone increasing (decreasing) for all $x \in (0,\infty)$. If $T =T_n$ for some $n \in \bN$, then $\sigma(\cdot,x)_{|_{[T_{n-1},T_n]}}$ must be decreasing $\forall x \in (0,\infty)$ and $\sigma(\cdot,x)_{|_{[T_{n},T_{n+1}]}}$ increasing $\forall x \in (0,\infty)$, as $T \notin B^{\sigma}$. Hence, $\underline{\sigma}(x):=\sigma(T,x)$ is a lower bound on $[T_{n-1}, T_{n+1}]$. In any case, $\sigma$ satisfies conditions \eqref{assum:integralfinite} and \eqref{assum:smallerthansigma}, thus Lemma \ref{lem:sigmatolocalmartingale} implies that $X$ is a strict local martingale. 

We now prove that if \eqref{assum:integralinfinityforfixedT} holds, then $X$ is a martingale. For the sake of simplicity, we only consider the case when $\sigma(\cdot,x)_{|_{[0,T_1]}}$ is increasing $\forall x \in (0,\infty)$ and $\sigma(\cdot,x)_{|_{[T_{1},\infty)}}$ is decreasing $\forall x \in (0,\infty)$. The proof can be analogously generalized.

Under this assumption, we have that \eqref{assum:integralinfinityforfixedT} is equivalent to
 \begin{align}
&\int_c^{\infty}  \frac{x}{\sigma^2(T,x)}dx = \infty \text{ for any $T>0$, $T\ne T_1$.}
\end{align}
 
We get
\begin{equation}\label{eq:stbiggerT1}
{\bE^Q}[X_t|\F_s] =  {\bE^Q}[X^{s,X_s}_t|\F_s]=X_s, \quad T_1<s<t,
\end{equation}
where the last equality comes from Lemma  \ref{lemma:sigmatomartingalesgen}, choosing $r=s$ and $\overline{\sigma}^{(T)}(x) := \sigma(s,x)$ for any $T>s$. 
We now show that the same holds for $s=T_1$. %Theorem 2.2 of \cite{ekstrom2012dupire} implies that the function $t \to {\bE^Q}[X_t]$ is continuous, so that in particular 
 Theorem 2.2 of \cite{ekstrom2012dupire} states that the functions $C, P:[0, \infty) \times (0, \infty) \to \reals^+$ defined by
 $$
 C(T,K):= {\bE^Q}[(X_T-K)^+], \quad  P(T,K):= {\bE^Q}[(K-X_T)^+], \qquad (T,K) \in [0,\infty) \times  (0, \infty),
 $$ 
 are continuous with respect to $T$ for any fixed $K>0$.  From this and from the put-call parity formula
 $$
 C(T,K) =  P(T,K) - K +   {\bE^Q}[X_{T}],  \qquad (T,K) \in [0,\infty) \times  (0, \infty),
 $$
 it follows that
 $$
  \lim_{t \downarrow T_1} {\bE^Q}[X_t] =  \lim_{t \downarrow T_1}C(t,K)- \lim_{t \downarrow T_1}P(t,K) + K = C(T_1,K)-P(T_1,K) + K = {\bE^Q}[X_{T_1}].
 $$
Since in our case ${\bE^Q}[X_t]$ is constant for any $t>T_1$ by  \eqref{eq:stbiggerT1}, we get ${\bE^Q}[X_t]={\bE^Q}[X_{T_1}]  $ for any $t>T_1$.  This implies that 
\begin{equation}\label{eq:sequalT1}
{\bE^Q}[X_t|\F_{T_1}] =  X_{T_1}, \quad T_1<t,
\end{equation}
because $X$ is a supermartingale.
Moreover, again by Lemma  \ref{lemma:sigmatomartingalesgen} and using that $\sigma(\cdot,x)_{|_{[0,T_1]}}$ is increasing $\forall x \in (0,\infty)$, we get
\begin{equation}\label{eq:stsmallerT1}
{\bE^Q}[X_t|\F_s] = X_s, \quad s<t<T_1.
\end{equation}
Finally, we have that
\begin{equation}\label{eq:T1betweensandt}
{\bE^Q}[X_t|\F_s] = {\bE^Q}\left[{\bE^Q}[X_t|\F_{T_1}] | \F_s \right] = {\bE^Q}[X_{T_1}|\F_s] = X_s, \quad 0 < s < T_1 \le t,
\end{equation}
where the second equality follows by \eqref{eq:sequalT1}. The last equality holds because $X$ is a supermartingale and ${\bE^Q}[X_{T_1}] =  {\bE^Q}[X_{s}]$ for all $s \le T_1$, since the expectation is continuous and  constant by \eqref{eq:stsmallerT1} for any $s<T_1$.

Putting together equations \eqref{eq:stbiggerT1}, \eqref{eq:sequalT1}, \eqref{eq:stsmallerT1} and \eqref{eq:T1betweensandt} we get that $X$ is a true martingale.

\end{proof}

The results above show that, under some conditions, the behaviour of the volatility for large values of the spatial variable determines if the process $X$ is a strict local martingale, i.e., if it has bubbly dynamics.

In the sequel we use these findings as well as a modified version of Dupire's equation in  \cite{ekstrom2012dupire} in order to {prove the existence of a ``bubble detection function'' in local volatilities models}. 

 \subsubsection{Neural network approximation for local volatility models}\label{sec:neuralapproxlocal}
{We now provide a neural network approximation also for the case $\alpha=1$. To this purpose we first }
 state  the general Dupire equation for call options for local volatility models, see Theorem 2.3 of \cite{ekstrom2012dupire}. 
\begin{theorem}\label{thm:dupireformulageneralized}
Let $X$ be the process in \eqref{eq:localvolatility}. The price $C^{\alpha, X}(T,K)$ in \eqref{eq:pricecallcollateralized} of a call option on $X$ with strike $K>0$ and maturity $T>0$ with collateral requirement $\alpha \in [0,1]$ is the unique bounded classical solution $C^{\alpha, X} \in C^{1,2}\left((0,\infty) \times (0,\infty)  \right) \cap C\left([0,\infty) \times [0,\infty)  \right)$ of the equation
\begin{equation}\label{eq:dupire}
\begin{cases}
\partial_TC^{\alpha, X}= \mathcal{L} C ^{\alpha, X} - (1-\alpha)\partial_Tm \qquad \text{for $(T,K) \in (0, \infty) \times (0, \infty)$} \\
C^{\alpha, X}(0,K)=(x-K)^+ \\
C^{\alpha, X}(T,0)= \bE^Q[X_T]+\alpha m(T),
\end{cases}
\end{equation}
where  $\mathcal{L}$ is the second order differential operator 
$$
\mathcal{L}=\frac{\sigma^2(T,K)}{2} \partial_{KK}
$$
and $m$ is the martingale defect of $X$ introduced in \eqref{eq:martingaledefect}. %Moreover, $m_T(\cdot)$ denotes the derivative of $m(\cdot)$ with respect to time $T$.
\end{theorem}

Note that, since the martingale defect $m$ is increasing with respect to time if $X$ is a strict local martingale, then $ (1-\alpha)\partial_Tm>0$ for $\alpha<1$.

In particular, since  
$$
C^{\alpha, X}(T,0)= \bE^Q[X_T]+\alpha  \left( X_0 - \bE^Q[X_T]\right) = \alpha X_0 + (1-\alpha) \bE^Q[X_T],
$$
%which implies
%$$
% \bE^Q[X_T] = \frac{C^{\alpha, X}(T,0)- \alpha X_0}{1-\alpha}
%$$
if $\alpha<1$ we can write 
$$
\partial_Tm(T)=\partial_T \left( X_0 - \bE^Q[X_T]\right) = -\frac{\partial_T C^{\alpha, X}(T,0)}{1-\alpha}.
$$
Therefore, for any $\alpha \in [0,1]$ it holds
$$
\partial_TC^{\alpha, X}(T,K)= \mathcal{L} C ^{\alpha, X}(T,K) - (1-\alpha)\partial_Tm(T,K)= \frac{\sigma^2(T,K)}{2} \partial_{KK} C ^{\alpha, X}(T,K) + \partial_T C^{\alpha, X}(T,0),
$$
from which we get
\begin{equation}\label{eq:dupiresigma}
\sigma(T, K) = \sqrt{\frac{ 2\partial_{T}\left(C^{\alpha, X}(T,K)-C^{\alpha, X}(T,0)\right)}{ \partial_{KK} {C^{\alpha, X}} (T,K)}}
\end{equation}

 if $\partial_{KK}{C^{\alpha, X}}(T,K)>0$ for any $(T,K) \in  [0,\infty) \times (0,\infty)$. Namely, the PDE in \eqref{eq:dupire} allows to recover the function $\sigma$ from prices of (possibly collateralized) call options.  Together with Theorem \ref{thm:sigmatomartingales}, we use the PDE \eqref{eq:dupire} and \eqref{eq:dupiresigma} for detecting asset price bubbles by looking at prices of call options, by means of neural network approximations. We formalize this in the following.
 
We now assume that call prices are given by \eqref{eq:pricecallcollateralized} with $\alpha =1$, which is the case that  {was not covered by the results in Section \ref{sec:neuralapprox}}. We denote by $C^X$ the price given in such a way, with underlying $X$ defined by  \eqref{eq:localvolatility}, where $\sigma$ satisfies Assumption \ref{ass:monotone} and
\begin{equation}\label{eq:Sigma}
 \partial_{KK}{C^{X}}(T,K)>0 \text{ for any $(T,K) \in  [0,\infty) \times (0,\infty)$}.
 \end{equation}
 \black{Since $C^{X}(T,0)=X_0$, equation \eqref{eq:dupiresigma} becomes the classical Dupire formula
 \begin{equation}\label{eq:classicaldupiresigma}
\sigma(T, K) = \sqrt{\frac{ 2\partial_{T}C^{X}(T,K)}{ \partial_{KK} {C^{X}} (T,K)}}.
\end{equation}
 }
  We prove an equivalent of Proposition \ref{prop:existenceF} in the current setting. %An analogous result can be proved with slight variations also in the case $\alpha \in [0,1)$.

%Our aim is to prove the existence of a sequence of neural network maps $(\hat{F}^n)_{n \in \bN}$, $\hat{F}^n \colon (0,\infty)^n \to [0,1]$, see Definition \ref{def:neuralnetwork}, with the following property: if $x$ is a vector of call option prices written on an underlying $X$ satisfying \eqref{eq:localvolatility}, with collateral constant $\alpha$ and maturities and strikes $(T_i,K_i)$, $i=1,\dots,n$, chosen well enough, then $\hat{F}^n(x) \approx 1$ if $X$ is a strict local martingale under the pricing measure $Q$ and $\hat{F}^n(x) \approx 0$ otherwise, with accuracy increasing with $n$ and tending to $100\%$ for $n \to \infty$. Once such a sequence is given, we can use $\hat{F}^n$ to detect bubbles: if we are given observed market prices of call options $\hat{C}(T_i,K_i)$, $i=1,\dots, n$, we evaluate  $\hat{F}^n(\hat{C}(T_1,K_1),\ldots,\hat{C}(T_n,K_n))$ to detect if there is a bubble or not. 

%From now on, we assume without loss of generality that the option prices are given by equation \eqref{eq:pricecallcollateralized} with $\alpha = 1$.  
%
Introduce the spaces
$$
\M_{loc}^{locvol}= \{\text{$X=(X_t)_{t \ge 0}$ with dynamics \eqref{eq:localvolatility}, where $\sigma$ satisfies Assumption \ref{ass:monotone} and \eqref{eq:Sigma}}\},
$$
and
\begin{equation}\notag
\M_{L}^{locvol}= \{\text{$X\in \M_{loc}^{locvol}$, $X$  strict local martingale}\}.
\end{equation}
Consider the set 
\begin{equation}\label{eq:mathcalXlocvol}
\X^{locvol} := \{C^X \,:\, X \in \M_{loc}^{locvol} \} \subset C^{1,2}([0,\infty) \times (0,\infty),(0,\infty))
\end{equation}
endowed with the Borel sigma-algebra of the Fr\'echet space $C^{1,2}([0,\infty) \times (0,\infty),(0,\infty))$ with topology induced by the countable family of seminorms
\begin{equation}\label{eq:seminormsC12}
\|f\|_{i,j,n} := \sup  \left\{|\partial^{(i,j)}_{t,x}f(t,x)|, \text{ } (t,x) \in [0,n] \times [1/n, n] \right\},
\end{equation}
for $n \in \bN$ and $(i,j)= (0,0), (0,1), (0,2),(1,0).$

\begin{proposition}\label{prop:existenceFlocal}
{Under Assumption \ref{ass:monotone}}  there exists a measurable function $F \colon \X^{locvol} \to \{0,1\}$ such that $F(C^X) = \mathbbm{1}_{\{X \in \M_{L}^{locvol}\}}$. In particular,  
 \begin{equation}\label{eq:Ftodectedbubble2}
F(f)\black{=}\mathbbm{1}_{\left\{\text{ there exists $t \in (0,\infty) \setminus B^{(\partial_t f, \partial_{xx} f)}$ such that $\int_{c}^{\infty} x\frac{\partial_{xx} f(t,x)}{\partial_{t} f(t,x)}dx < \infty$} \right\}}, \quad f \in{\X^{locvol}},
\end{equation}
for arbitrary fixed $c>0$ and with 
\begin{equation}\label{eq:Bf}
B^{(\partial_t f, \partial_{xx} f)}:=\left\{T>0: \text{$\sigma:=\sqrt{\frac{\black{2}\partial_t f}{\partial_{xx}f}}$ has a local maximum at $T$ for all $x \in (0,\infty)$}\right\}.
\end{equation}
\end{proposition}

\begin{proof}
We have that ${\X^{locvol}} \subset C^{1,2}([0,\infty) \times (0,\infty),(0,\infty))$ by \black{Theorem \ref{thm:dupireformulageneralized}}. 
Fix $c>0$, and consider $X \in \M_{loc}^{locvol}$ with $C^X \in {\X^{locvol}}$. 
% Since $\sigma$ is strictly positive and continuous, 
 Theorem \ref{thm:Xmartingaleifconditionforanyt} implies that {for any process $X$ of the form \eqref{eq:localvolatility} where the local volatility function $\sigma$ satisfies Assumption \ref{ass:monotone} it holds}
\begin{equation}\label{eq:alternativeformulationforfixedc}
\text{$X \in \M_{L}^{locvol}$ if and only if $\exists$ $t\in (0,\infty) \setminus B^{\sigma}: \int_c^{\infty} \frac{x}{\sigma^2(t,x)}dx < \infty$,}
\end{equation}
with $B^{\sigma}$ in \eqref{eq:Bsigma}. Consider now $F \colon{\X^{locvol}} \to \{0,1\}$ as in \eqref{eq:Ftodectedbubble2}. 
From \eqref{eq:classicaldupiresigma} and \eqref{eq:alternativeformulationforfixedc}, it is clear that $F(C^X) = \mathbbm{1}_{\{X \in {\M^{locvol}_L\}}}$ for any $C^{X} \in {\X^{locvol}}$.
% Introduce $\tilde F \colon \X \to \Sigma$ by $\tilde F(f)=\sigma$ with
%$$
%\sigma(x):= \sqrt{\frac{1}{2}\frac{f_t(t,x)}{f_{xx}(t,x)}}, \qquad x \in (0,\infty),
%$$
%for a given $t >0$. Define then $\hat F:\Sigma \to \{0,1\}$ by
%$$
%\hat F(\sigma) := \mathbbm{1}_{\{\text{there exists $c>0$ such that  $\int_{c}^{\infty} \frac{x}{\sigma^2(x)}dx < \infty$}\}}, \qquad \sigma \in \Sigma,
%$$
%and call $F := \hat F \circ \tilde F$. It is clear that 
We now prove that $F$ is measurable. {Fix a volatility function $\sigma$ satisfying Assumption \ref{ass:monotone}}. First note that, if 
\begin{equation}\label{eq:considerQ}
D^{\sigma, \mathbb{Q}}:=\left(\left\{t >0:  \int_c^{\infty} \frac{x}{\sigma^2(t,x)}dx < \infty \right\} \cap \mathbb{Q}\right) \setminus B^{\sigma} = \emptyset,
\end{equation}
then
$$
D^{\sigma}:=\left\{t >0 :  \int_c^{\infty} \frac{x}{\sigma^2(t,x)}dx < \infty \right\} \setminus  B^{\sigma}  = \emptyset.
$$
Indeed, suppose $D^{\sigma} \ne \emptyset$ and consider $t \in D^{\sigma}$. Then either $t$ is a local minimum for $\sigma(\cdot, x)$ for all $x \in (0,\infty)$ or there exists an interval $I_t$ with $t \in I_t$ on which $\sigma(\cdot, x)$ is monotone increasing for all $x \in (0,\infty)$ or monotone  decreasing for all $x \in (0,\infty)$: in all cases, it is possible to find a positive $\bar t \in \mathbb{Q} \setminus B^{\sigma}$ such that $ \int_c^{\infty} \frac{x}{\sigma^2(\bar t,x)}dx < \infty$,  hence $D^{\sigma, \mathbb{Q}} \ne \emptyset$.

Then we can write 
\begin{equation}\label{eq:Falpha1}
F(f)= 1- \prod_{t \in \mathbb{Q} \cap \left((0,\infty) \setminus B^{(\partial_t f, \partial_{xx} f)} \right)}(1-F_{t}(f)) = 1- \prod_{t \in \mathbb{Q} \cap (0,\infty)} (1-F_{t}(f)\mathbb{1}_{\left\{t \notin B^{(\partial_t f, \partial_{xx} f)}\right\}}), 
\end{equation}
$f \in {\X^{locvol}}$, with $B^{(\partial_t f, \partial_{xx} f)}$ as in \eqref{eq:Bf} and where $F_t \colon {\X^{locvol}} \to \{0,1\}$ is defined by
\begin{equation}\label{eq:Ftalpha1}
F_t(f):=\mathbbm{1}_{\left\{\int_{c}^{\infty} x\frac{\partial_{xx} f(t,x)}{\partial_{t} f(t,x)}dx < \infty \right\}}, \quad f \in {\X^{locvol}},
\end{equation}
for any $t>0$. Note that for any $t \in \mathbb{Q} \cap (0,\infty)$ and any $\delta \in (0,t)$ we have
$$
\mathbb{1}_{\left\{t \notin B^{(\partial_t f, \partial_{xx} f)}\right\}} = 1 -  \prod_{x  \in \mathbb{Q} \cap (0,\infty)}  \prod_{s  \in \mathbb{Q} \cap (0,\delta)}  \mathbb{1}_{\left\{\frac{\partial_t f(t-s,x)}{\partial_{xx}f(t-s,x)}<\frac{\partial_t f(t,x)}{\partial_{xx}f(t,x)}\right\}}\mathbb{1}_{\left\{\frac{\partial_t f(t+s,x)}{\partial_{xx}f(t+s,x)}<\frac{\partial_t f(t,x)}{\partial_{xx}f(t,x)}\right\}},
$$
which is measurable. Moreover, we can write $F_t = \hat F \circ \tilde{F}_t$, where the operator
 $\tilde F_t \colon {\X^{locvol}} \to C\left((0,\infty), (0,\infty) \right)$ is given by
$$
\tilde F_t (f):= I_d \cdot \frac{\partial_{xx} f(t, \cdot)}{\partial_t f(t,\cdot)}
$$
for any $t>0$, with $I_d(x)=x$, and $\hat F:C\left((0,\infty), (0,\infty) \right) \to \{0,1\}$ equal to
$$
\hat F (g) = \mathbbm{1}_{\{\int_c^{\infty} g(x) dx < \infty \} }.
$$
The space $C\left((0,\infty), (0,\infty) \right)$ is a Fr\'echet space equipped with the Borel sigma-algebra with respect to the topology induced by the family of seminorms 
\begin{equation}\label{eq:seminormsC}
\|g\|_{m} := \sup  \left\{|g(x)|, \text{ } x \in [1/m,m] \right\}
\end{equation}
for $m \in \mathbb{N}$. Note that $\tilde{F}_t$ is well defined for any $t>0$ by \eqref{eq:Sigma} and by \eqref{eq:classicaldupiresigma}, since $\sigma$ is strictly positive. Moreover, $\tilde{F}_t$ is measurable for all $t>0$ because it is continuous with respect to the uniform convergence on compact sets away from zero induced by the topologies of the Fr\'echet spaces $C^{1,2}([0,\infty) \times (0,\infty),(0,\infty))$ on ${\X^{locvol}}$ and $C\left((0,\infty), (0,\infty) \right)$ with the seminorms in \eqref{eq:seminormsC12} and \eqref{eq:seminormsC}, respectively. \\
Coming now to $\hat F$, it can be written as 
$$
\hat F(g)=  \mathbbm{1}_{\{I(g) < \infty \} }
$$
where $I : C\left((0,\infty), (0,\infty) \right) \to [0,\infty]$ is the integral operator
$$
I(g) = \int_c^\infty g(x)dx, \quad g \in C\left((0,\infty), (0,\infty) \right).
$$
Clearly, $I = \lim_{k \to \infty} I_k$, where the functions $I_k: C\left((0,\infty), (0,\infty) \right) \to [0,\infty)$ with
$$
I_k(g) :=  \int_c^k g(x)dx, \quad g \in C\left((0,\infty), (0,\infty) \right), \quad k \in \bN,
$$
are measurable because the operator $g \to \int_c^{k} g(x) dx$ is continuous for any finite $k$. Then $I$ is measurable and $\hat F$ as well.

The function $F_t$ in \eqref{eq:Ftalpha1} is thus measurable  for every $t>0$ since it is the composition of measurable functions. Hence, the function 
$F$ is measurable by \eqref{eq:Falpha1}. 

%$$
%\hat F(\sigma) := \mathbbm{1}_{\{\text{there exists $c>0$ such that  $\int_{c}^{\infty} \frac{x}{\sigma^2(x)}dx < \infty$}\}}, \qquad \sigma \in \Sigma,
%$
\end{proof}

\black{
\begin{remark}
In Proposition \ref{prop:existenceFlocal}, we focus on $\alpha=1$ because the case $\alpha<1$ has already been covered in Section \ref{sec:neuralapprox}. Alternatively, if $\alpha<1$, using equation \eqref{eq:dupiresigma} we can define a new detection function  $F \colon \X^{locvol} \to \{0,1\}$ similar to \eqref{eq:Ftodectedbubble2} whose measurability can be proved using the same arguments as in the proof of Proposition \ref{prop:existenceFlocal}.
\end{remark}
}

\begin{remark}
The assumption that $ \partial_{KK} {C^{X}} (T,K)>0$ for any $(T,K) \in [0,\infty) \times (0,\infty)$ in \eqref{eq:Sigma}  is equivalent to 
$$
\partial_T {C^{X}} (T,K)>0, \quad (T,K) \in [0,\infty) \times (0,\infty),
$$
which is usually satisfied for the observed market prices. 

The above conditions are satisfied, for instance, if $X_T$ admits a strictly positive density $g_T^X : \reals_+ \to \reals_+$ for any $T \ge 0$. Indeed, under this assumption it holds that
$$
 \partial_{KK} {C^{X}}(T,K) =  \partial_{KK}  \int_K^{\infty} (x-K) g_T^X(x) dx = g_T^X(K) > 0
$$
for any $K>0$. 
\end{remark}

We now give a result {analogous} to Proposition \ref{prop:existenceapproximation}. The proof is follows the same steps.
\begin{proposition}\label{prop:existenceapproximationlocvol}
 Define the set
\begin{equation}\notag
A=[0,\infty) \times [C,\infty) 
\end{equation}
for some fixed $C>0$. Let $(T_n, K_n)_{n \in \bN}$  be a sequence of maturities and strikes such that the set $ \{(T_n, K_n), n \in \bN\}$ is dense in $A$. Let $\mu$ be a probability measure on the space $\X^{locvol}$ introduced in \eqref{eq:mathcalXlocvol} such that there exists a probability measure $\nu$ on ${\X^{locvol}}$ with
	\begin{equation}\label{eq:measurenotzerolocvol}
	\nu\left(\left\{f\in \X^{locvol} \text{ such that } f(T_i, K_i) = g(T_i,K_i)\text{ for all $i \in \bN$}\right\}\right)>0 \quad \text{for any $g \in \supp(\mu)$}.
	\end{equation}
Then there exists a sequence of functions $(F^n)_{n \in \bN}$, $F^n \colon \reals^n \to [0,1]$ such that%which converges to $F$ pointwise, i.e., such that for any fixed $\sigma \in \Sigma$ it holds
\begin{equation}\notag
\int_{\X^{locvol}} \left| F^n(g(T_1,K_1),\ldots,g(T_n,K_n))-F(g) \right|^pd\mu(g)  \xrightarrow[n \to \infty]{}0,
\end{equation}
for any $p \in [1,\infty)$, where $F \colon \X^{locvol} \to \{0,1\}$  is the function introduced in {\eqref{eq:Ftodectedbubble2}}.
%In particular, given a probability measure $\mu$ on $(\Sigma,\B^{\Sigma})$, for any $\varepsilon >0$ there exists $n_0 \in \bN$ such that for all $n \geq n_0$ we have 
%\begin{equation}\label{eq:resultconvergence}
%\int_{\Sigma}|F^n(C^\sigma(T_1,K_1),\ldots,C^\sigma(T_n,K_n))  - F(C^\sigma)| \mu(d \sigma) <\varepsilon.
%\end{equation}
For fixed $n \in \mathbb{N}$, the function $F^n$ can be chosen as follows:
%	\begin{equation}\label{eq:Fn}
%	F^n(c):= \frac{1}{|S_n(c)|} \sum_{\tilde \sigma \in S_n(c)} F\left(C^{\tilde{\sigma}}\right)
%	\end{equation}
%	where
%	\begin{equation}\label{eq:Snc}
%	S_n(c) := \left\{ \tilde \sigma \in \Sigma \text{ such that } C^{\tilde{\sigma}}(T_i, K_i) = c_i \text{ for all $i = 1, \dots, n$}\right\}.
%	\end{equation}
\begin{align}\notag
	F^n(c^n):= \begin{cases}
	\frac{1}{\nu\left(S_n(c^n)\right)} \int_{S_n(c^n)}F(f)d\nu(f) \quad &\text{if $\nu\left(S_n(c^n)\right) > 0$},\\
	0 \quad &\text{otherwise},
	\end{cases}
	\end{align}
		$c^n \in \reals^n$, where
	\begin{equation}\label{eq:Snc}
	S_n(c^n) := \left\{f\in \X^{locvol} \text{ such that } f(T_i, K_i) = c^n_i \text{ for all $i = 1, \dots, n$}\right\}, \quad c^n \in \reals^n.
	\end{equation}
\end{proposition}

We then have the following result. The proof follows the same steps as for Theorem \ref{theo:existencenetwork}.

\begin{theorem}\label{theo:existencenetworklocvol}
 Fix $p\in [1,\infty)$  and let $(T_n, K_n)_{n \in \bN}$  be the sequence of maturities and strikes introduced in Proposition \ref{prop:existenceapproximationlocvol}. Also let $\mu$ be a probability measure on $\X^{locvol}$ such that there exists a probability measure $\nu$ on $\X$ which satisfies \eqref{eq:measurenotzerolocvol}. Then for any $\epsilon >0$ there exists an $n \in \bN$ and a neural network $\hat{F}^n \colon \reals^n \to [0,1]$ such that
$$
\int_{\X^{locvol}} |\hat{F}^n(g(T_1,K_1),...,g(T_n,K_n)) - F(g)|^p {d\mu(g)} < \epsilon.
$$
% Fix $n \in \bN$, $p>0$ and let $F$ be the function defined in \eqref{eq:Ftodectedbubble1} and \eqref{eq:Ftodectedbubble2}. Then for any $\epsilon>0$ there exists a neural network $\hat{F}^n \colon \reals^{n} \to \{0,1\}$ such that
% $\lVert F - \hat F^n \rVert_p \le \epsilon$, where $\lVert \cdot \rVert_p$ is the $L^p$-norm.
\end{theorem}

 \section{Numerical experiments}\label{sec:numexperiments}
 \subsection{The methodology}\label{sec:nummethodology}
 %The method illustrated in Section \ref{sec:theoretical} does not require any parameter estimation related to the asset price, and it is model independent when $\alpha<1$. 
 {In Section \ref{sec:theoretical} we have provided a theoretical foundation for the bubble detection methodology explained in the introduction. The method does not require any parameter estimation related to the asset price, and it is model independent when $\alpha<1$.} Fixing $\alpha=1$, we now show via numerical experiments that our approach works not only within a class  of models, but also if the network is trained using a certain class of stochastic processes (like local volatility models) and tested within another class (for example, stochastic volatility models). \black{In Section \ref{sec:marketdata}}, we also test the method with market data associated to assets involved in the new tech bubble burst at the beginning of 2022 (see among others \cite{Bercovici}, \cite{Ozimek}, \cite{Seria}, \cite{Sharma}) finding a close match between the output of the network and the expected results.

{\begin{remark}
We show the results for a collateral constant $\alpha=1$ because this is the most common choice in the literature, see for example \cite{dias2020note},  \cite{heston2007options} and \cite{JarrowProtter2010}. %Moreover, in this way we show that our methodology also works in the case not covered by the results in Section \ref{sec:neuralapprox}. 
For completeness, we have repeated the same tests for $\alpha=\frac{1}{2}$ finding not inferior accuracies.
\end{remark}}

\begin{remark}
The theoretical results proved in Section \ref{sec:theoretical} provide an important justification for our approach: they show that in our setting bubbles can indeed be detected using a neural network with call option price inputs up to an error that tends to $0$ as the number of call option prices and the size of the neural network \blue{become} arbitrarily large. The functional form of the bubble indicator $F$ in \eqref{eq:Ftodectedbubble1} and \eqref{eq:Ftodectedbubble2} is not explicitly used in the numerical algorithm, but instead an approximation of $F$ is learnt from data. An alternative approach for bubble detection could aim to directly approximate the functional form of $F$ in \eqref{eq:Ftodectedbubble1} and \eqref{eq:Ftodectedbubble2} via an extrapolation of the calibrated local volatility function as in \cite{jarrow2011detect}. One of the key features of our approach is that no choice of extrapolation needs to be made explicitly, but instead learning is only based on the available call option prices. 
\end{remark}

\black{For the sake of clarity, before showing our results we specify how our methodology works:
\begin{itemize}
\item We first fix a constant $x_0 \in \mathbb{R}$ and a set of strikes and maturities $(T_j,K_j)_{1 \le j \le m}$, $m \in \bN$.
\item We take a set of local martingales $(X^i)_{1 \le i \le n_{train}}$, $n_{train} \in \bN$, with possibly different dynamics and such that $X^i_0=x_0$ for any $1 \le i \le n_{train}$. Half of them are true martingales, half strict local martingales.
\item We consider a binary classification problem with $m$-dimensional input and output in $\{0,1\}$. Each of the $n_{train}$ input data points consists of the call option prices with underlying $X^i$ for strikes and maturities $(T_j,K_j)_{1 \le j \le m}$. The label associated to data point $i$ is the indicator $\mathbb{1}_{\{X^i\text{ is a strict local martingale}\}}$. The call option prices are computed either via Monte-Carlo techniques or, when possible, analytically.
\item Based on these data points we train a neural network in order to make it learn the indicator function from the prices. Since this is a supervised learning, binary classification problem, we choose binary cross entropy loss function and Sigmoid activation function for the output layer. 
\end{itemize}
We test the network on a set of new local martingales $(Y_i)_{1 \le i \le n_{test}}$, $n_{test} \in \bN$, with $Y^i_0=x_0$ for any $1 \le i \le n_{test}$, again with the property that half of them are strict local martingales and half true martingales. In particular, our choice for the output layer's activation function implies that, for any stock,  the network returns a number $p\in(0,1)$ which can be interpreted as the probability that the underlying under consideration is a strict local martingale. If $p > 0.5$, we label the underlying to be a strict local martingale according to the network. As performance metrics, both on the training and test set, we consider \emph{accuracy} $a$ 
and \emph{$F_1$ score}, defined as
$$
a:=\frac{\text{true positives}+\text{true negatives}}{\text{true positives}+\text{true negatives}+\text{false positives}+\text{false negatives}}
$$
and
$$
F_1:=2\frac{\text{precision}\cdot\text{recall}}{\text{precision}+\text{recall}},
$$
respectively, where
$$
\text{precision}:=\frac{\text{true positives}}{\text{true positives}+\text{false positives}}, \quad \text{recall}:=\frac{\text{true positives}}{\text{true positives}+\text{false negatives}}.
$$
In our framework, \emph{true positives} (\emph{true negatives}) are strict local martingales (true martingales) correctly labelled as strict local martingales (true martingales), and \emph{false positives} (\emph{false negatives}) are true martingales (strict local martingales) uncorrectly labelled as strict local martingales (true martingales).
\\
For both accuracy and $F_1$ score, and for all the cases under investigation, we compute the average over $10$ runs of the network's training, for fixed training and test data \footnote{\black{Since the training of the network is not deterministic, different outputs might be obtained for different runs.}}.} 

\black{%In Sections \ref{sec:localvol} and \ref{sec:stochvol}, 
Unless differently specified, we consider a neural network having a first layer with $60$ nodes, two hidden layers with $30$ and $7$ nodes, respectively, and ReLu activation function for the input and hidden layers. We use an Adam optimization algorithm with learning rate equal to $0.001$ and $20$ epochs. These choices strike a balance between accuracy and efficiency, enabling the algorithm to yield high performance indicators in a significant number of cases without excessive computational overhead. In Section \ref{sec:stochvol}, where the stochastic models in the training and test set are significantly different, we find that other architectures and choices of the learning rate deliver better performances in some cases. %We always mention any modifications to the choices above.
}

 \subsection{Local volatility models}\label{sec:localvol}
% We now test the approach described in Section \ref{sec:theoretical}. Call prices are here considered fixing $\alpha=1$ in \eqref{eq:pricecallcollateralized} as it is the most used pricing formula under bubble, see for example \cite{JarrowProtter2010}.
% As a first experiment, 
 First, we consider a displaced CEV process $(X_t)_{t \ge 0}$ with dynamics
 \begin{equation}\label{eq:CEV}
 dX_t = \sigma \cdot (X_t+d)^{\beta} dW_t, \quad t \ge 0, \quad X_0 = x_0 > 0,
 \end{equation}
where $(W_t)_{t \ge 0}$ is a one-dimensional standard Brownian motion and $\sigma, \beta, d >0$.

\begin{proposition}\label{prop:CEV}
The SDE \eqref{eq:CEV} has a unique strong solution, which is a true martingale if and only if $\beta \le 1$.  
\end{proposition}
\begin{proof}
When $d=0$, the SDE \eqref{eq:CEV} admits a unique strong solution since the function $\sigma(x)=x^{\beta}$, $x > 0$, satisfies Assumption \ref{ass:sigma}. A process $X$ with dynamics given by \eqref{eq:CEV} for general $d>0$ is of the form $X_t = Y_t - d$, $t \ge 0$, where $Y$ is the unique strong solution of \eqref{eq:CEV} with $d=0$ and $Y_0=X_0+d$. By Theorem \ref{thm:Xmartingaleifconditionforanyt}, $Y$ is a true martingale if and only $\beta \le 1$, and thus the same holds for $X$.
\end{proof}
The proof of the above proposition shows that all the results of Section \ref{sec:theoretical} also apply to \eqref{eq:CEV}, even if $X$ is in general not positive but only bounded from below. \black{Based on this result, we proceed to show our numerical experiments in Sections \ref{sec:randomdisplacement} and \ref{sec:fixeddisplacement}}.

\subsubsection{Random displacement}\label{sec:randomdisplacement}
As a first experiment, we work with $n_{train} \in \bN$ different underlyings, all following \eqref{eq:CEV} with fixed $\sigma>0$ and different, randomly chosen parameters $d$ and $\beta$. In particular, for any $i =1, \dots, n_{train}$, the underlying $X^i$ has displacement $d_i$ uniformly distributed in an interval $[0,D]$ for given $D>0$ and strictly positive exponent $\beta_i$. We let $\beta_i$ be uniformly distributed in an interval $(a,1]$ with $a \in (0,1)$ in $n_{train}/2$ cases and in an interval $(1,A)$ with $A>1$ for the other $n_{train}/2$ cases, so that we consider $n_{train}/2$ true martingales and $n_{train}/2$ strict local martingales. 
We also fix strikes $0<K_1<\dots<K_m$ and maturities $0<T_1<\dots<T_l$ for $l, m \in \mathbb{N}$. 
Note that  the values of the calls can be computed analytically, as indicated for example in \cite{lindsay2012simulation}. 

We choose $n_{train} =100000$, $n_{test}=5000$, $a=0.2$, $A=2$, $D=0.2$, $\sigma=1$, $X_0=2$, $100$ equally spaced maturities from $T_1=1$ to $T_l=2$. We consider two sets of $51$ equally spaced strikes: one from $K_1=1.8$ to $K_m=2.3$ (``at-the-money strikes'') and one \blue{from $K_1=3.5$ to $K_m=4$} ( ``out-of-the-money strikes''). For strikes in $[1.8, 2.3]$, we obtain:
\begin{itemize}
\item Accuracies of $99.1\%$ for both the training and test data; 
\item $F_1$ scores of $0.992$ and $0.991$ for the training and test data, respectively; 
\end{itemize}
For strikes in $[3.5, 4]$, we obtain:
\begin{itemize}
\item Accuracies of $99.5\%$  and  $99.4\%$ for the training and test data, respectively; 
\item $F_1$ scores of $0.995$ and $0.994$ for the training and test data, respectively.
\end{itemize}

\black{As a further experiment, we give the network the same matrix as before but where the last element of the $i$-th row is now given by the martingale defect of $X^i$, $1 \le i \le n_{train}$. In this case, the network has to learn the martingale defect from the option prices. We choose linear activation function for the output layer and the $R^2$ coefficient of determination score of the computation of the martingale defect as a performance indicator.} This is given by 
$$
R^2:=1-\frac{\sum_{i=1}^{n_{test}}  (y_i-f_i)^2}{\sum_{i=1}^{n_{test}} (y_i-\bar y)^2},
$$
where $y_i$ and $f_i$, $i=1,2, \dots, n$, are the true and predicted values, respectively, and $\bar y$ is the average of the true values. The best possible score is therefore 1. 
 \black{In our case, on the test set we get an $R^2$ coefficient of determination score of $0.994$ for strikes in $[1.8, 2.3]$ and of $0.999$ for strikes in $[3.5, 4]$.}

\black{All the metrics show better results for strikes out-of-the-money. A possible explanation is the following: looking at equation \eqref{eq:pricecallcollateralized} we note that the price of a call option in presence of a bubble diverges from the one without bubble by a term proportional to the martingale defect $m(T) = X_0 - \bE^Q[X_T]$, and this term has a big impact on the price of the option when $C(T,K)$ in  \eqref{eq:pricecallcollateralized} is small, i.e., for large values of the strike $K$. }

\subsubsection{Train with no displacement, test with a strictly positive displacement}\label{sec:fixeddisplacement}
We now repeat the experiment described in Section \ref{sec:randomdisplacement} with the difference that this time the displacements are not random, and they differ from training and testing. In particular, we train and test the network with data generated by an underlying in \eqref{eq:CEV} with $d=0$ and $d=0.5$, respectively. In this way, we can check if the network is able to detect bubbles also when it is trained and tested with data coming from processes having shifted dynamics. All the other parameters are the same as in Section \ref{sec:randomdisplacement}. \black{For strikes in $[1.8, 2.3]$, we obtain:
\begin{itemize}
\item Accuracies of $99.9\%$ and $86.9\%$ for the training and test data, respectively; 
\item $F_1$ scores of $0.999$ and $0.883$ for the training and test data, respectively; 
\end{itemize}
For strikes \blue{in} $[3.5, 4]$ we get:
\begin{itemize}
\item Accuracies of $99.9\%$ and $91.1\%$ for the training and test data, respectively; 
\item $F_1$ scores of $0.999$ and $0.918$ for the training and test data, respectively; 
\end{itemize}}
\black{As before, strikes out-of-the-money yield better results. However, we observe smaller accuracies and $F_1$-score in the test set} than in the first case: this can be expected because of the shifted dynamics of the underlying processes that produce the prices to train and test the network, respectively. Note that  the displacement we use to test the network is bigger than the biggest possible displacement in the first experiment.

\black{Very similar results are observed if we  instead train the network with displacement $d=0.5$ and test with $d=0$.} %we observe the following:

In all the experiments, the accuracy of the results gets smaller if we decrease the number of strikes or the number of maturities, keeping the intervals the same.

\subsection{Stochastic volatility models}\label{sec:stochvol}

As an extension of the setting of Section \ref{sec:localvol}, we now include stochastic volatility models into our numerical experiments. In particular, we check how our methodology
works when the network is fed with data coming from stochastic volatility processes and tested looking at local volatility models, and vice-versa. Since the nature of
the two classes of models is deeply different, this analysis constitutes a robust test for our methodology.

In particular, we consider the SABR model and the class of stochastic volatility models introduced in \cite{sin1998complications}. The SABR volatility model $F=(F_t)_{t \ge 0}$ is defined as the unique strong solution of the SDE
\begin{align}
dF_t &= \sigma_t (F_t)^{\gamma} dW_t, \qquad t \ge 0, \label{eq:F} \\
d\sigma_t &= \alpha \sigma_t dZ_t, \qquad t \ge 0, \label{eq:sigma}
 \end{align}
with initial value $F_0>0$, initial volatility $\sigma_0>0$, exponent $\gamma>0$ and volatility of the volatility $\alpha>0$. Here $W$ and $Z$ are Brownian motions with correlation $\rho \in [-1,1]$. 

 A class of stochastic volatility models of the form 
\begin{align}
&dY_t=\sigma_1v_t^{\alpha}Y_tdB^1_t + \sigma_2v^{\alpha}_tY_tdB_t^2, \quad t \ge 0, \label{eq:S} \\
&dv_t=a_1v_tdB^1_t + a_2v_tdB_t^2 + \kappa(L-v_t)dt, \quad t \ge 0, \quad v_0 = 1,\label{eq:v}
\end{align}
is considered in \cite{sin1998complications}, where $\alpha, \kappa, L \in \mathbb{R}^+$, $\sigma_1, \sigma_2, a_1, a_2 \in \mathbb{R}$ \black{and $(B_1,B_2)$ is a two-dimensional Brownian motion}. It is well known that $Y$ is a strict local martingale if and only if $a_1 \sigma_1 + a_2 \sigma_2 > 0$, see  \cite{sin1998complications}. From now on, we call  \eqref{eq:S}-\eqref{eq:v} the Sin model. It can be noted that, when $\gamma=1$, the SABR model is a particular case of the Sin model, and is a strict local martingale if and only if $\rho >0$. 
 
 The call prices for the SABR model admit an analytical approximation\footnote{We compute these approximations using the library available at \\
 \hyperlink{https://github.com/google/tf-quant-finance/tree/master/tf_quant_finance/models/sabr}{\texttt{https://github.com/google/tf-quant-finance/tree/master/tf$\_$quant$\_$finance/models/sabr}}..}, see for example \cite{piiroinen2018asset}.
 For the Sin model, on the other hand, we are not aware of any way to analytically compute or approximate the prices, so we rely on Monte-Carlo approximations with $200000$ simulations and time step $0.01$.
 
We  take $61$ maturities equally spaced from $T_1=2$ to $T_l=5$ and, as in Section \ref{sec:localvol}, two groups of $100$ equally spaced strikes: one from $K_1=1$ to $K_m=3.5$ and one from $K_1=3$ to $K_m=5.5$. We consider three different classes of processes for the underlying, in all cases with initial value $X_0=2$:
\begin{itemize}
\item A set of displaced CEV models with same specifications as in Section \ref{sec:localvol}, that is, in particular, with randomly varying exponent and displacement;
\item A set of SABR models \eqref{eq:F}-\eqref{eq:sigma} with parameters $\alpha=0.5$, $\gamma=1$ and $\sigma_0=0.5$. We let the correlation $\rho$ between the process and the volatility be randomly chosen and a priori different for any underlying: in particular, we set it to be uniformly distributed in $[-0.8,0]$ in half of the cases (the ones where the underlying is a true martingale) and in $(0, 0.8]$ in the other half of the cases (the ones where the underlying is a strict local martingale).
\item A set of Sin models \eqref{eq:S}-\eqref{eq:v} with parameters $v_0=0.5$, $\alpha=1$, $\kappa=0$, $\sigma_2=-0.5$, $a_1=1.8$ and $a_2=1.2$. We let $\sigma_1$ be randomly chosen and a priori different for any underlying: specifically, it is uniformly distributed in  $[0,1/3]$ in half of the cases (the ones where the underlying is a true martingale) and in $(1/3, 1]$ in the other half of the cases (the ones where the underlying is a strict local martingale).
\end{itemize}
\black{For all the experiments, we consider $50000$ underlyings for the training set and  $5000$ for the test set, and  investigate accuracy and $F_1$ score of our methodology.}

As a benchmark for our following investigation, for which training and test sets will consist of prices generated by underlying coming from different classes of models, Table \ref{tab:accuracysingle} shows the accuracy \black{and the $F_1$ score} in the training and test phases when we both train and test the neural network with any of the three models above. Half of the them are true martingales and half strict local martingales. 
\begin{table}[h]
 \centering
\black{ \begin{tabular}{|c|c|c|c|}
 \hline
&  Displaced CEV & Sin & SABR \\
 \hline
Strikes in $[1,3.5]$ & \makecell{$a(\text{train})=99.8\%$, \\ $a(\text{test})=98.1\%$  \\ $F_1(\text{train})=0.988$ \\ $F_1(\text{test})=0.927$} &\makecell{$a(\text{train})=99.8\%$, \\ $a(\text{test})=99.6\%$  \\ $F_1(\text{train})=0.995$ \\ $F_1(\text{test})=0.934$} & \makecell{$a(\text{train})=99.4\%$, \\ $a(\text{test})=94.8\%$  \\ $F_1(\text{train})=0.997$ \\ $F_1(\text{test})=0.876$}  \\
 \hline
Strikes in $[3,5.5]$& \makecell{$a(\text{train})=99.9\%$ \\ $a(\text{test})=99.2\%$ \\ $F_1(\text{train})=0.991$ \\ $F_1(\text{test})=0.941$} &\makecell{$a(\text{train})=99.8\%$, \\ $a(\text{test})=99.7\%$ \\ $F_1(\text{train})=0.998$ \\ $F_1(\text{test})=0.947$} & \makecell{$a(\text{train})=99.3\%$ \\ $a(\text{test})=98.4\%$ \\ $F_1(\text{train})=0.995$ \\ $F_1(\text{test})=0.914$}  \\
\hline
\end{tabular}}
\caption{\black{Accuracy ($a(\text{train})$ and $a(\text{test})$) and $F_1$ score ($F_1(\text{train})$ and $F_1(\text{test})$) on the training set and on the test set}   when the algorithm to identify bubbles is both trained and tested with one single family of models: the three investigated cases are  displaced CEV, Sin model and SABR model. }
\label{tab:accuracysingle}
\end{table}

\black{We now mix the models for training and testing. Since training and test data are now drawn from very different distributions, in some cases using the same network architecture and learning rate specified in Section \ref{sec:nummethodology} is not optimal, and better results can be obtained with different specifications. In particular, we let the learning rate $\lambda$ vary in $\{0.01, 0.001, 0.0001\}$, and also consider a bigger network constituted by an input layer with $100$ nodes, two hidden layers with $70$ and $30$ nodes, and possibly $L1$ or $L2$-regularization with parameter in the set  $\{0.1, 0.01, 0.001\}$. For some choices of the training and test data, all reported in Tables \ref{tab:accuracytestcev} and \ref{tab:accuracytestsabr}, we find that changing the learning rate or taking the bigger, regularized network lead to bigger accuracy and $F_1$ score in the test set. All the results shown in Tables \ref{tab:accuracytestcev} and \ref{tab:accuracytestsabr} are relative to the best choice of network and learning rate. When the choice of the network and of the learning rate is different from the specifications in Section \ref{sec:nummethodology}, we indicate it. In particular, when we write  ``Big, Reg'' this means that we use the bigger network with $L2$-regularization with parameter $0.01$: this turns out to be the best regularization choice in terms of performance in the test set in all cases when regularization is required. Looking at Tables \ref{tab:accuracytestcev} and \ref{tab:accuracytestsabr}, it can be noted that only two variants of the setting described in Section \ref{sec:nummethodology}  are needed to always get the highest test accuracy and $F_1$ score: the bigger, regularized network and the ``default" network  with smaller learning rate.}

First, to investigate the performance of our method when trained on stochastic volatility processes and tested on local volatility models, we test the algorithm looking at the prices generated by the set of displaced CEV models, for three different training sets: one constituted exclusively from data coming from the SABR model, one only from the Sin model and one from both of them. When both the Sin and the SABR models are used for training, the training set is equally split between the two cases. Table \ref{tab:accuracytestcev} shows the accuracy \black{and the $F_1$ score} we get on the training set and on the test set in the different cases.
%\begin{table}[h]
%\centering
% \black{\begin{tabular}{|c|c|c|c|}
% \hline
%&  Only SABR & Only Sin & Both SABR and Sin \\
% \hline
%Strikes in $[1,3.5]$ & \makecell{$a(\text{train})=97.2\%$ \\ $a(\text{test})=67.2\%$  \\ $F_1(\text{train})=0.953$ \\ $F_1(\text{test})=0.639$} &\makecell{$a(\text{train})=99.6\%$ \\ $a(\text{test})=50.0\%$  \\ $F_1(\text{train})=0.939$ \\ $F_1(\text{test})=0.334$} & \makecell{$a(\text{train})=93.5\%$ \\ $a(\text{test})=75.0\%$  \\ $F_1(\text{train})=0.868$ \\ $F_1(\text{test})=0.655$}  \\
% \hline
%Strikes in $[3,5.5]$& \makecell{$a(\text{train})=99.8\%$ \\ $a(\text{test})=85.8\%$ \\ $F_1(\text{train})=0.997$ \\ $F_1(\text{test})=0.875$ } &\makecell{$a(\text{train})=99.7\%$ \\ $a(\text{test})=78.3\%$  \\ $F_1(\text{train})=0.958$ \\ $F_1(\text{test})=0.759$} & \makecell{$a(\text{train})=98.0\%$ \\ $a(\text{test})=83.2\%$  \\ $F_1(\text{train})=0.934$ \\ $F_1(\text{test})=0.806$}  \\
%\hline
%\end{tabular}}
%\caption{\black{Accuracy ($a(\text{train})$ and $a(\text{test})$) and $F_1$ score ($F_1(\text{train})$ and $F_1(\text{test})$) on the training set and on the test set} when the algorithm to identify bubbles is tested looking at call prices generated by displaced CEV models, for two different sets of strikes and three different training sets.}\label{tab:accuracytestcev}
%\end{table}

\begin{table}[h]
\centering
 \black{\begin{tabular}{|c|c|c|c|}
 \hline
&  Only SABR & Only Sin (Big, Reg) & SABR and Sin (Big, Reg)   \\
 \hline
Strikes in $[1,3.5]$ & \makecell{$a(\text{train})=99.8\%$ \\ $a(\text{test})=68.7\%$  \\ $F_1(\text{train})=0.998$ \\ $F_1(\text{test})=0.761$} & \makecell{$a(\text{train})=99.7\%$ \\ $a(\text{test})=82.8\%$  \\ $F_1(\text{train})=0.997$ \\ $F_1(\text{test})=0.790$}& \makecell{$a(\text{train})=98.7\%$ \\ $a(\text{test})=95.2\%$  \\ $F_1(\text{train})=0.986$ \\ $F_1(\text{test})=0.955$}   \\
 \hline
Strikes in $[3,5.5]$& \makecell{$a(\text{train})=99.9\%$ \\ $a(\text{test})=86.2\%$ \\ $F_1(\text{train})=0.999$ \\ $F_1(\text{test})=0.875$ } &\makecell{$a(\text{train})=99.9\%$ \\ $a(\text{test})=86.3\%$  \\ $F_1(\text{train})=0.999$ \\ $F_1(\text{test})=0.841$} & \makecell{$a(\text{train})=97.3\%$ \\ $a(\text{test})=86.7\%$ \\ $F_1(\text{train})=0.977$ \\ $F_1(\text{test})=0.879$} \\
\hline
\end{tabular}}
\caption{\black{Accuracy ($a(\text{train})$ and $a(\text{test})$) and $F_1$ score ($F_1(\text{train})$ and $F_1(\text{test})$) on the training set and on the test set} when the algorithm to identify bubbles is tested looking at call prices generated by displaced CEV models, for two different sets of strikes and three different training sets.}\label{tab:accuracytestcev}
\end{table}

We see that looking at prices of out-of-the-money options \black{generally} improves the performance, as in Section \ref{sec:localvol}, \black{apart from the case when the network is trained with prices generated by both SABR and Sin  processes}. Moreover, splitting the training set into the two sets of SABR and Sin processes results in an higher accuracy \black{and $F_1$ score} on the test set with respect to both cases, when we train with only one set of processes, at the expenses of the performance on the training set.

We then move to test our methodology looking at the prices generated by the class of SABR models, again for three different training sets: the first two sets consist of call prices generated by the displaced CEV models and by the Sin models, respectively, and the third one of prices generated by both of them. In this way, we complete our investigation: the first case regards training with local volatility and testing with stochastic volatility, the second one concerns both training and testing with stochastic volatility, whereas the third one studies if a combination of training from local and stochastic volatility can be beneficial when looking at a stochastic volatility model. The results are summarized in Table \ref{tab:accuracytestsabr}.
%\begin{table}[h]
%\centering
%\black{ \begin{tabular}{|c|c|c|c|}
% \hline
%&  Only displaced CEV & Only Sin & Both displaced CEV and Sin \\
% \hline
%Strikes in $[1,3.5]$ & \makecell{$a(\text{train})=99.1\%$ \\ $a(\text{test})=50.0\%$  \\ $F_1(\text{train})=0.906$ \\ $F_1(\text{test})=0.365$} &\makecell{$a(\text{train})=99.4\%$ \\ $a(\text{test})=68.0\%$  \\ $F_1(\text{train})=0.954$ \\ $F_1(\text{test})=0.615$} & \makecell{$a(\text{train})=97.3\%$ \\ $a(\text{test})=92.2\%$ \\ $F_1(\text{train})=0.838$ \\ $F_1(\text{test})=0.775$}  \\
% \hline
%Strikes in $[3,5.5]$& \makecell{$a(\text{train})=99.8\%$ \\ $a(\text{test})=63.4\%$  \\ $F_1(\text{train})=0.998$ \\ $F_1(\text{test})=0.423$} &\makecell{$a(\text{train})=99.8\%$ \\ $a(\text{test})=91.3\%$  \\ $F_1(\text{train})=0.998$ \\ $F_1(\text{test})=0.923$} & \makecell{$a(\text{train})=99.7\%$ \\ $a(\text{test})=96.8\%$ \\ $F_1(\text{train})=0.997$ \\ $F_1(\text{test})=0.952$}  \\
%\hline
%\end{tabular}}
%\caption{\black{Accuracy ($a(\text{train})$ and $a(\text{test})$) and $F_1$ score ($F_1(\text{train})$ and $F_1(\text{test})$) on the training set and on the test set}  when the algorithm to identify bubbles is tested looking at data coming from the SABR model, for two different sets of strikes and three different training sets.}\label{tab:accuracytestsabr}
%\end{table}

\begin{table}[h]
\centering
\black{ \begin{tabular}{|c|c|c|c|}
 \hline
&  Only disp. CEV (Big, reg) & Only Sin & Disp. CEV and Sin ($\lambda = 0.0001$) \\
 \hline
 Strikes in $[1,3.5]$& \makecell{$a(\text{train})=99.8\%$ \\ $a(\text{test})=50.0\%$  \\ $F_1(\text{train})=0.998$ \\ $F_1(\text{test})=0$} &\makecell{$a(\text{train})=99.9\%$ \\ $a(\text{test})=73.0\%$  \\ $F_1(\text{train})=0.999$ \\ $F_1(\text{test})=0.787$} & \makecell{$a(\text{train})=99.8\%$ \\ $a(\text{test})=88.5\%$ \\ $F_1(\text{train})=0.997$ \\ $F_1(\text{test})=0.898$}  \\
 \hline
 Strikes in $[3,5.5]$ & \makecell{$a(\text{train})=98.1\%$ \\ $a(\text{test})=79.2\%$  \\ $F_1(\text{train})=0.982$ \\ $F_1(\text{test})=0.736$} &\makecell{$a(\text{train})=99.4\%$ \\ $a(\text{test})=93.4\%$  \\ $F_1(\text{train})=0.997$ \\ $F_1(\text{test})=0.930$} & \makecell{$a(\text{train})=99.7\%$ \\ $a(\text{test})=96.0\%$ \\ $F_1(\text{train})=0.997$ \\ $F_1(\text{test})=0.959$}  \\
\hline
\end{tabular}}
\caption{\black{Accuracy ($a(\text{train})$ and $a(\text{test})$) and $F_1$ score ($F_1(\text{train})$ and $F_1(\text{test})$) on the training set and on the test set}  when the algorithm to identify bubbles is tested looking at data coming from the SABR model, for two different sets of strikes and three different training sets.}\label{tab:accuracytestsabr}
\end{table}

Looking at bigger strikes  we get bigger accuracy \black{and $F_1$-score in the test set}  in all cases. % \black{except for the test set when we train only with the displaced CEV model}. 
Moreover, for both the groups of strikes, training with a combination of data coming from a stochastic volatility model and from a local volatility model increases \black{both accuracy and $F_1$-score} in the test set for the SABR model at the expenses of the accuracy in the training set: this suggests again that training with higher variety of data makes our methodology  more robust. On the other hand, \black{for the at-the-money-strikes}, training only with data coming from the displaced CEV processes does not help guessing the strict local martingale property of the SABR processes: in particular, the algorithm always identifies all the processes as \black{true} martingales. Comparing this case with Table \ref{tab:accuracytestcev} \black{for this choice of strikes} we see that our methodology works fairly good if we train with the SABR models and try to label the CEV models, but not vice-versa.

\black{In the Appendix we show the Receiver Operating Characteristic (ROC) curves for all the cases taken into consideration in Tables \ref{tab:accuracytestcev} and \ref{tab:accuracytestsabr}. ROC curves plot the true positive rate (TPR) against the false positive rate (FPR) at various values of the acceptance threshold, see for example \cite{carter2016roc}, \cite{fan2006understanding}, \cite{gonccalves2014roc} for a deeper analysis. In particular, Figures \ref{fig:rocat} and \ref{fig:rocout} display the ROC curves for the at-the-money and the out-of-the-money strikes, respectively.}

\black{We note that, when we train with CEV and test with SABR for at-the-money strikes, the ROC curve is just the diagonal line: this is due to the fact that the network always \blue{predicts} true martingales in the test set. }

\black{Moreover, in seven cases, the ROC curve is (very close to) the combination of the vertical line at zero and the upper line in the plot. For these examples, the network exhibits two distinct patterns:
\begin{enumerate}[label=\alph*)]
\item  It consistently identifies strict local martingales correctly: for any strict local martingale in the test set, it outputs a probability close to $1$ that it is a strict local martingale. This implies that $\text{TPR}=1$ for all thresholds smaller than a value proximal to $1$. In particular, the thresholds for which $\text{TPR} <1$ are big enough so that $\text{FPR}=0$. 
\item  It consistently identifies true martingales correctly: for any true martingale in the test set, it outputs a probability close to $0$ that it is a strict local martingale. This implies that $\text{FPR}=0$  for all thresholds bigger than a value close to $0$. In particular, the thresholds for which $\text{FPR}>0$ are small enough so that $\text{TPR}=1$ .
\end{enumerate}
In both the scenarios above, for all thresholds we have either $\text{TPR}=1$ or $\text{FPR}=0$, and this explains the observed ROC curves. In particular, scenario a) occurs when:
\begin{itemize}
\item we train with CEV and Sin processes and test with SABR (for both groups of strikes);
\item we train with Sin and test with SABR (for out-of-the-money strikes);
\item we train with SABR and Sin and test with CEV (for at-the-money strikes).
\end{itemize}
Scenario b) instead is obtained when:
\begin{itemize}
\item we train with CEV and test with SABR (for out-of-the-money strikes);
\item we train with Sin and test with CEV (for both groups of strikes).
\end{itemize}}

%
%So: . And this explains these strange ROC curves you observe, which are (or are very close to) the y axes and the upper line in the plot: the tpr rate is 1 for all thresholds for which the fpr is > 1. 
%

\black{Finally, in order to show the role that epochs play during training, we investigate the evolution of the accuracy in the set constituted by the prices generated by SABR and Sin processes, with same parameters specified before, as a function of the epoch number. In this case, a fraction of the training data is used as validation data. The model sets apart this fraction of the training data, does not train on it, and evaluates the accuracy on this data at the end of each epoch. In Figure \ref{epochs} we show the training set accuracy (that is, the accuracy on the data that are actually used to train the model) and validation set accuracy letting the validation set size be $10\%$ of the total training set size. }

\black{\begin{figure}
\centering
 \includegraphics[scale=0.45]{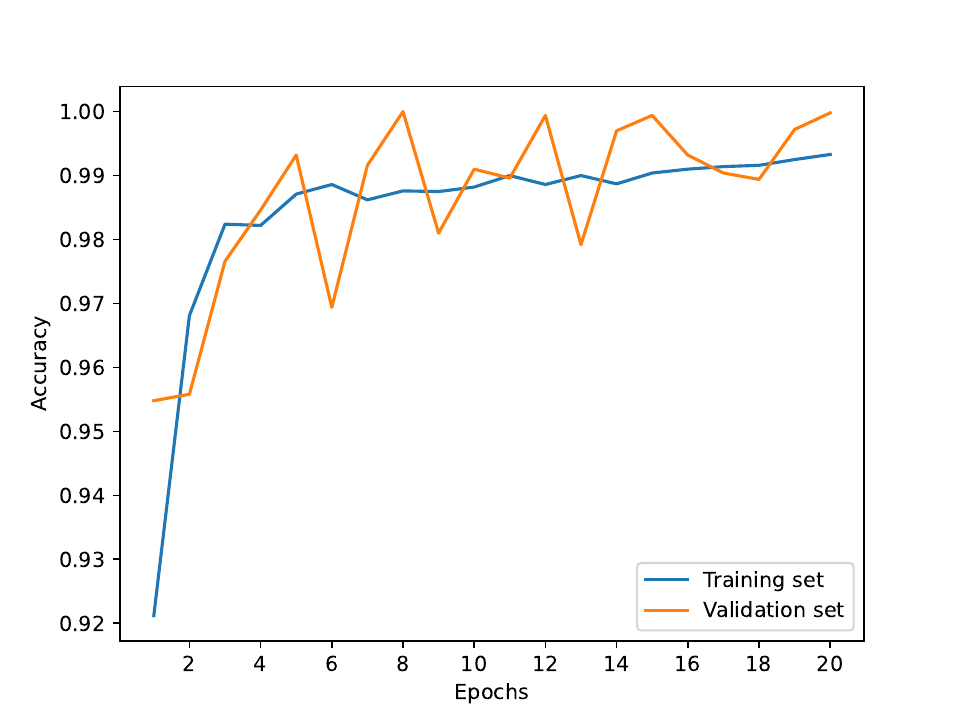}
 \caption{Training set accuracy and validation set accuracy with respect to epochs, when training with SABR and Sin models. The validation set size is $10\%$ of the total training set size.}
 \label{epochs}
\end{figure}}

\black{Summarizing}, the above results show that, in particular when combining different classes of processes in the training phase, our methodology still works when including stochastic volatility models, even if the results are (not surprisingly) less accurate than when we consider local volatility models for both training and testing.

\subsection{The choice of the network architecture}\label{sec:netarch}
We conclude Section \ref{sec:numexperiments} by commenting on the chosen neural network architecture. The network described in Section \ref{sec:nummethodology} (which we call network $A$ in the remainder of the current subsection) has been selected since it guarantees good performances for most of cases, for a fairly small amount of layers and nodes.

In order to test the sensitivity of our approach to the choice of the network architecture, we also consider other possible choices, \blue{initially keeping other hyperparameters such as learning rate and number of epochs unchanged. In particular, we introduce}:
\begin{itemize}
\item Network $B$: same number of layers as $A$  but more nodes:  input layer with $100$ nodes, two hidden layers with $70$ and $30$ nodes;
\item Network $C$:  same number of nodes and layers as $B$, but with $L2$ regularization with parameter $0.01$;
\item Network $D$:  higher number of layers and more nodes than $A$:  input layer with $200$ nodes, four hidden layers with $150$, $100$, $70$ and $30$ nodes.  
\end{itemize}
We then repeat the experiments of Sections \ref{sec:localvol} and \ref{sec:stochvol} for any of these networks. We note that:
\begin{itemize}
\item \blue{The average training accuracy is above $99.5\%$ for all networks and all cases;}
\item Networks $B$ and $D$ never give (strictly) the best results in terms of either accuracy or $F_1$ score in the test set. In particular, choosing $D$ or $B$ decreases the performance in the test set when we train with Sin and test with SABR, train with  Sin and CEV and test with SABR, train with Sin and test with CEV, train with Sin  and SABR and test with CEV.  In all other cases (among which the local volatility models of Section \ref{sec:localvol}) accuracies and $F_1$ score in the testing set remain unchanged. 
\item The regularized network $C$ is (strictly) the best choice in terms of performance in the test set when we train with CEV and test with SABR, when we train with Sin and SABR and test with CEV and when we train with Sin and test with CEV (see also Tables \ref{tab:accuracytestcev} and \ref{tab:accuracytestsabr}); in particular, in the first case and for out-of-the-money strikes, it represents a big improvement with respect to the non regularized network $B$, which gives a test accuracy of only $54\%$.
For the examples of Section  \ref{sec:localvol}, instead, regularization yields worse results in terms of both  \blue{test} accuracy and $F_1$ score.
\item Network $A$ strictly gives the best results in the test set when we train with Sin and test with SABR, train with Sin and CEV and test with SABR, and when we train with SABR and test with CEV.
\end{itemize}
\blue{For network $D$, which has the highest number of parameters to optimize, we also evaluate whether increasing the number of epochs from 20 to 40 or 60 improves the average testing accuracy. We observe that training with more epochs significantly enhances performance only when training with Sin and testing with displaced CEV for at-the-money strikes. In this scenario, the average test accuracy rises from $80.2\%$ at 20 epochs to $84.1\%$ at 40 epochs. This latter value exceeds the accuracy reported for network C in Table \ref{tab:accuracytestcev}. However, since the results are similar and network D with $40$ epochs is computationally more costly, we retain network C for consistency.}

\section{Analysis on market data: detecting tech stocks bubbles}\label{sec:marketdata}
\subsection{\black{The methodology}}
We now use our \black{approach} to detect if some tech stocks have been affected by an asset price bubble in last years. In particular, we focus on \black{Meta Platforms Inc.\ (META), GameStop Corp.\ (GME),} Nvidia Corporation (NVDA) and Tesla Inc.\ (TSLA). These companies have experienced a boom in the last years, followed by a contraction at the beginning of 2022, as it can also be \black{seen by the time series of prices shown in Figures \ref{nvidia}, \ref{gamestop}, \ref{meta} and} \ref{tesla}. This has brought many financial analysts to claim the presence of a new tech bubble, after the \textit{dot com mania} of the late 1990s (see for example \cite{Bercovici},  \cite{fusari2020testing}, \cite{Ozimek}, \cite{Seria}, \cite{Sharma}). 

Let $S^i_{t_i}$ be the value at time $t_i $ of a stock $i \in \{\black{\text{META}, \text{ GME},} \text{ NVDA},  \text{ TSLA}\}$. In order to assess if a bubble is present at $t_i$, we consider the call option market prices $C^i_1, \dots, C^i_{n_i}$ at $t_i$ for maturities $T^i_1, \dots, T^i_{n_i}$ and strikes $K^i_1, \dots, K^i_{n_i}$.  \black{We use call option prices for a total of $50000$ different underlyings} for  maturities $T^i_1, \dots, T^i_{n_i}$,  strikes $K_1^i, \dots, K^i_{n_i}$ and initial value of the underlying $S^i$ \black{to train a neural network  which has a first layer with $60$ nodes, two hidden layers with $30$ and $7$ nodes, respectively, ReLu activation function for the input and hidden layers, Sigmoid activation function for the output layer and $L2$-regularization with parameter $0.01$. For training we use the Adam optimization algorithm with learning rate equal to $0.001$ and $20$ epochs.}   

 \black{As potential candidates to generate the prices used to train the network, we only consider displaced CEV models in \eqref{eq:CEV} and SABR models in \eqref{eq:F}-\eqref{eq:sigma} with $\gamma=1$, as they admit analytic formulas for the call prices. In particular:}
\begin{itemize}
\item For the SABR models, we let $\sigma_0$ and $\alpha$ be both uniformly distributed in $[0.1, 0.9]$. Moreover, we let $\rho$ be uniformly distributed in $[-0.9,0]$ in half cases (those where the underlying is a true martingale) and in $(0, 0.9]$ in the other half of cases (those where the underlying is a strict local martingale).
\item For the displaced CEV models, we let the volatility $\sigma$ and the displacement $d$ be uniformly distributed in $[0, 5]$ and $[0.8, 1.5]$, respectively. Moreover, we let the exponent $\beta$ be uniformly distributed in $[0.5,1]$ in half cases (those where the underlying is a true martingale) and in $(1,1.2]$ in the other half of cases (those where the underlying is a strict local martingale). 
\end{itemize}
Note that in the models specification above we take randomly chosen parameters in order to increase the variety of the training data and then the robustness of our methodology. 

% the highes. In most cases, the chosen underlyings are the displaced CEV models: contrary to the examples showed in Section \ref{sec:stochvol}, where the combination of local and stochastic volatility models enhances the accuracy on the test set, here for fewer and less deeply out-of-the-money strikes and maturities we observe a better performance for local volatility models alone. The network trained with the displaced CEV processes gives an accuracy of more than $90\%$ for most cases also on a training set constituted by calls with SABR underlyings}.

\black{After the training phase, we feed the network with the realized call market prices $C^i_1, \dots, C^i_{n_i}$, so that it outputs a number in $(0,1)$ which can be interpreted as the probability that the underlying has a bubble based on the training and on the market prices. We repeat this procedure $10$ times, and take the average of the probabilities returned by the networks.}

\black{For any date and any stock we decide whether to use  the network trained with data from displaced CEV models only, SABR models only or both, based on the following test. We train the network  for each of these alternatives, and for the maturities, strikes and initial values observed in the market. After running the prediction $10$ times with same training data, we select the training choice (displaced CEV only, SABR only or both) which ranks first in most of the metrics below: 
\begin{itemize}
\item Lowest variance of the output probabilities on the observed market data;
\item Highest average accuracy  on the training set;
\item Highest average accuracy on a test set constituted by option prices generated by both displaced CEV and SABR models;
\item Lowest variance of the output probabilities on the test set of displaced CEV and SABR models.
\end{itemize}
If two alternatives rank first in two of the above metrics, we take the average prediction between them. \\ According to this criterion we select $6$ times both displaced CEV and SABR models, $5$ times only SABR models and $15$ times only displaced CEV models. Moreover, in $7$ cases we take the average predictions between the network trained with both displaced CEV and SABR models and the one trained with  displaced CEV models only. In more detail:
\begin{itemize}
\item The network trained with both displaced CEV and SABR models gives  the lowest variance of returned probabilities for the market data in $14$ cases,  the highest accuracy on the training set in $1$ case,  the highest accuracy on the test set in $31$ cases and the lowest variance  of returned probabilities for the underlyings in the test set in $2$ cases;
\item The network trained with only SABR models gives  the lowest  variance of returned probabilities for the market underlying in $6$ cases and  the highest accuracy on the training set in $9$ cases;
\item The network trained with only displaced CEV models gives  the lowest variance of returned probabilities for the market underlying in $12$ cases,  the highest accuracy on the training set in $21$ cases,  the highest accuracy on the CEV/SABR test set in $2$ cases and the lowest variance  of returned probabilities for the underlyings in the test set in $30$ cases.
\end{itemize}
Note that, even if the training set with both displaced CEV and SABR models provides in most cases the highest accuracy on the test set, still training with the displaced CEV alone generally provides good test accuracies for the strikes and maturities available in the market. Moreover, for same strikes and maturities, it scores a lot better (with test accuracies around $97\%$) when the test set is generated by only displaced CEV models.
}

We now state our findings for the \black{four} underlyings. \black{Note that only American call option prices are available on the market. However, all the stocks that we consider here do not pay dividends or have a very small dividend yield (Nvidia Corporation) that we neglect. In this case, American call option prices are equal to European option prices and can be used for our tests.} 

\black{\subsection{Nvidia Corporation}}

The advanced semiconductor manufacturer Nvidia Corp. has seen its shares rocketing to huge heights in last years. From January 2016 to its peak in late November 2021, Nvidia's closing stock prices boosted their value by about $4600\%$. This seems to point to a bubble in Nvidia stock, see for example \cite{kolakowski}. 

We apply our methodology in order to check for the presence of a bubble from January 2016 to July 2022, on a six-monthly basis. In \black{Figure \ref{nvidia}} we plot the price evolution of the asset with the dates when we use our methodology to detect if a bubble is present.  In particular, for each date \black{we consider the average} probability $P_b$ for the presence of a bubble \black{we get running our methodology $10$ times, as described above}.  Different colors of the lines correspond to different values of $P_b$. We list these findings in more details in the following, in terms of the evolution of $P_b$. 

\begin{itemize}
\item No bubble is detected from January 2016 to \black{July 2017: in particular, $P_b$ is equal to $4.8 \%$ in January 2016,  $0.6 \%$ in July 2016, $8.7 \%$ in  January 2017 and $1.5\%$ in July 2017}. 

\item \black{The network outputs a probability of a bubble $P_b=71.2\%$ in January 2018  and $P_b=41.4\%$ in July 2018}.

\item \black{In the second half of 2018 we observe a decrease of the asset price, and the network outputs bubble probabilities $P_b=19.8\%$ in January 2019, $P_b=20.3\%$ in July 2019 and $P_b=41.2\%$ in January 2020.}

\item A second increase of the price leads to a second bubble according to our methodology. Indeed, the algorithm outputs \black{high probabilities of having a bubble from July 2020: we have $P_b=98.9\%$ in July 2020, $P_b=99.5\%$ in January and July 2021, $P_b=95.2\%$ slightly before the peak of November 2021}.

\item A second decline of the price can be observed from the end of 2021:  after the beginning of this decline, our methodology outputs \black{$P_b=44.3\%$}.

\end{itemize}

%We apply the algorithm described above to check if this is the case or not. 

%We consider four dates: 1st of June 2020, 1st of November 2021, 3rd of January 2022, 1st of February 2022. Note that the first two dates are before and the last two after the peak of the stock price evolution, see Figure \ref{nvidia}. Our methodology attributes a probability $100\%$ that a bubble is present in the first two dates and probabilities  $99.31\%$ and $99.98\%$, respectively, that no bubble is present the 3rd of January and the 1st of February. This is in line with the observed price trend, interpreting the contraction after the peak as the burst of the bubble. 

\black{\subsection{GameStop Corporation}}

\black{The stock value of the video games and gaming merchandise retailer GameStop has experienced a meteoric rise over the few last years. From the beginning of January 2021, GameStop's stock price increased from around $17\$$ to an all-time high of nearly $483\$$ on January 28, 2021, marking an astonishing gain of over $2700\%$ in a matter of weeks. After this surge, the price experienced a rapid decline, to rise again and stabilize, albeit with high volatility, at levels much higher than those of early 2021. Observing this behaviour, most financial analysts have claimed the stock clearly had a bubble, see for example \cite{libich2021bitcoin}. We here assess this statement using our methodology. Figure \ref{gamestop} shows the results. The network detects a bubble with high probability already in January and September 2020. In the middle of the huge rise of January 2021 a bubble is detected with average probability $P_b=86.7\%$, and later on always with probabilities higher than $99\%$ apart from June 2021 when the network outputs $P_b=94.7\%$. }

\black{\subsection{Meta Platforms Inc.}}
\black{The stock price of Meta Platforms, Inc., formerly known as Facebook, has experienced a substantial surge over the last few years, more than doubling from January 2018 to July 2021. We use our methodology in order to assess if this rise indicates the presence of a bubble, see Figure \ref{meta}. Before 2020, the network outputs small probabilities for Meta to have a bubble (average probabilities are $0.2\%$ in November 2018, $0.3\%$ in March 2019, $5\%$ in September 2019). Things change after the substantial rise from early 2020: $P_b$ equals $82.1\%$ in late July 2020, $92.8\%$ in early April 2021 and $93.4\%$ in September 2021. Finally, the decrease after the peak is interpreted as the burst of the bubble, as $P_b=0.7\%$ in November 2022. Note that the presence of temporary bubbles on Facebook over last years has been detected also with the method of  \cite{fusari2020testing}.}

\black{\subsection{Tesla Inc.}}
As a third example, we focus on the case of Tesla. The Tesla stock price has seen an extremely steep increase in last years, of about $30000\%$ from early 2016 to the peak in late 2021. Because of such a huge boost, many experts claimed that the stock is affected by a bubble, see for example \cite{libich2021bitcoin}. By using our methodology we  assess if a bubble is present or not four dates, which are shown in Figure \ref{tesla}. The color of the lines which identify the dates corresponds to the probability $P_b$ assigned by our methodology for the presence of a bubble. The algorithm outputs \black{$P_b=1.3\%$ the 3rd of January 2016},  \black{$P_b=70.1\%$ the 1st of November 2021 (right before the peak of the stock price), $P_b=95\%$ the 1st of February 2022 (right after the peak) and $P_b=98.9\%$ the 17th of July 2022}.

\black{
\subsection{Discussion}}
\black{Overall, for each of these assets our method finds that bubbles were present with high probability at certain points in time. These dates coincide with points shortly before a rapid decrease of the asset price. Similarly, points in time at which the method outputs low bubble probabilities coincide with times which are followed neither by a rapid increase nor a rapid decrease of the price.}

\begin{figure}
\centering
 \includegraphics[scale=0.95]{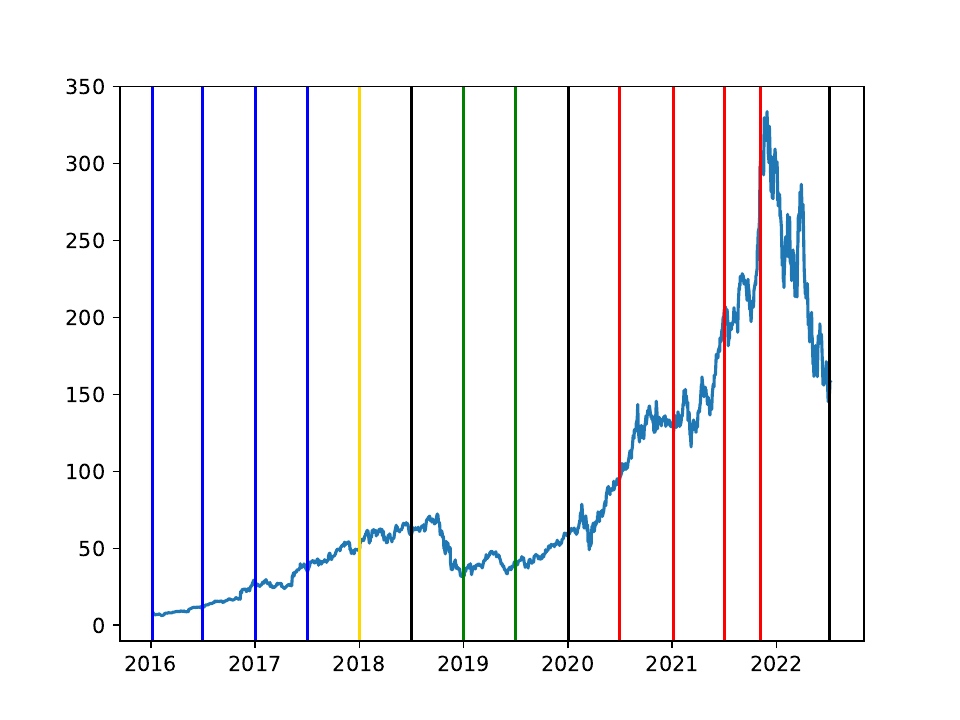}
 \caption{Nvidia Corporation stock prices (USD). The lines illustrate the dates considered to detect a bubble with our methodology. The colors change based on the \black{average} probability $P_b$ that the asset has a bubble according to our methodology. In particular: for the red lines we have $P_b>95\%$, for the \black{yellow} line \black{$P_b = 71.2\%$}, for the black lines \black{$40\% < P_b < 45\%$},  for the green lines $19\% < P_b < 21\%$ and for the blue lines $P_b \black{<10\%.}$}
 \label{nvidia}
\end{figure}
\begin{figure}
\centering
 \includegraphics[scale=0.7]{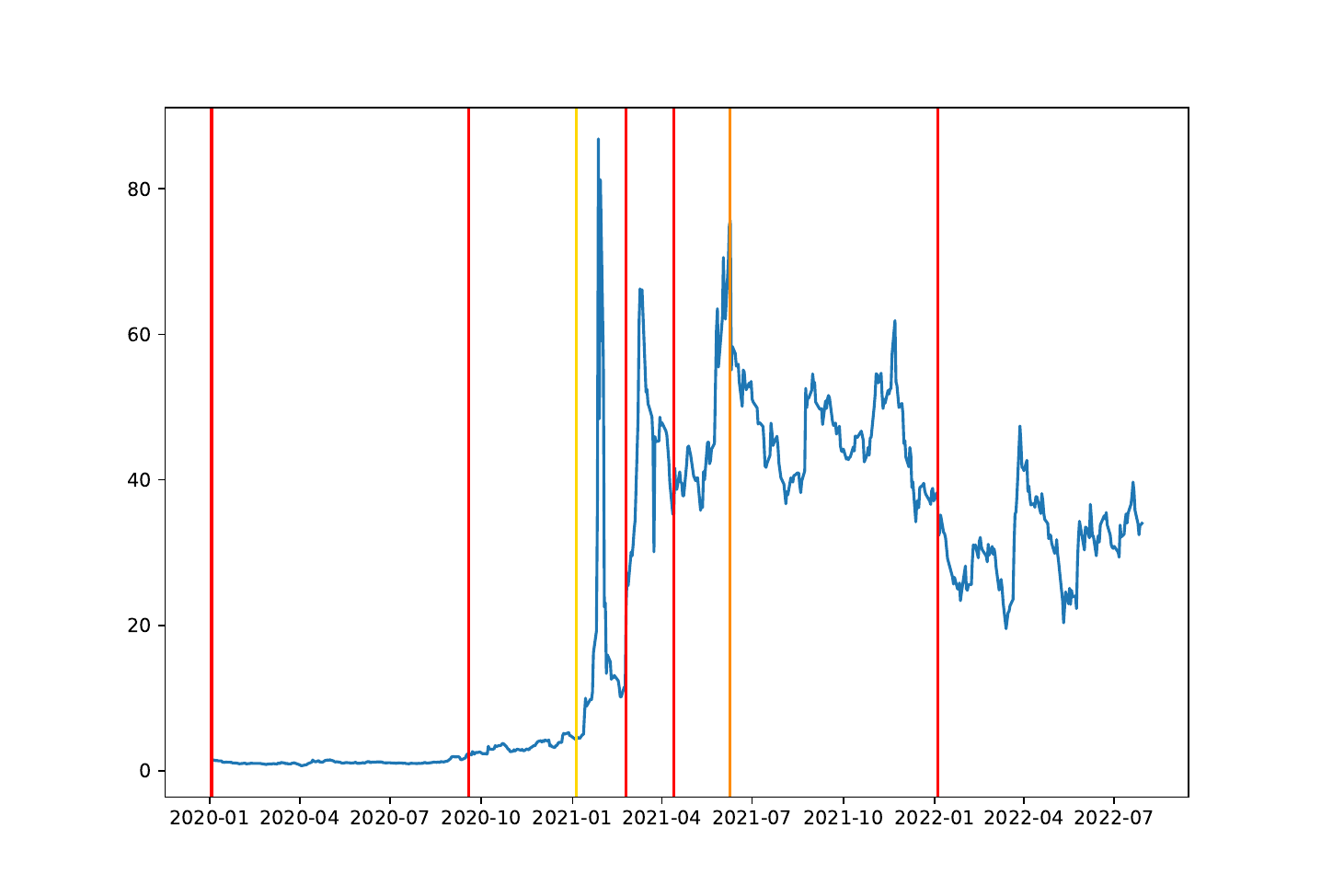}
 \caption{\black{GameStop corporation. stock prices (USD). The lines illustrate the dates considered to detect a bubble with our methodology. The colors change based on the average probability $P_b$ that the asset has a bubble according to our methodology. In particular, for the red lines we have \black{$P_b>98\%$}, for the  orange line $P_b=94.7\%$ and for the yellow one $P_b=86.7\%$.}}
 \label{gamestop}
\end{figure}
\begin{figure}
\centering
 \includegraphics[scale=0.75]{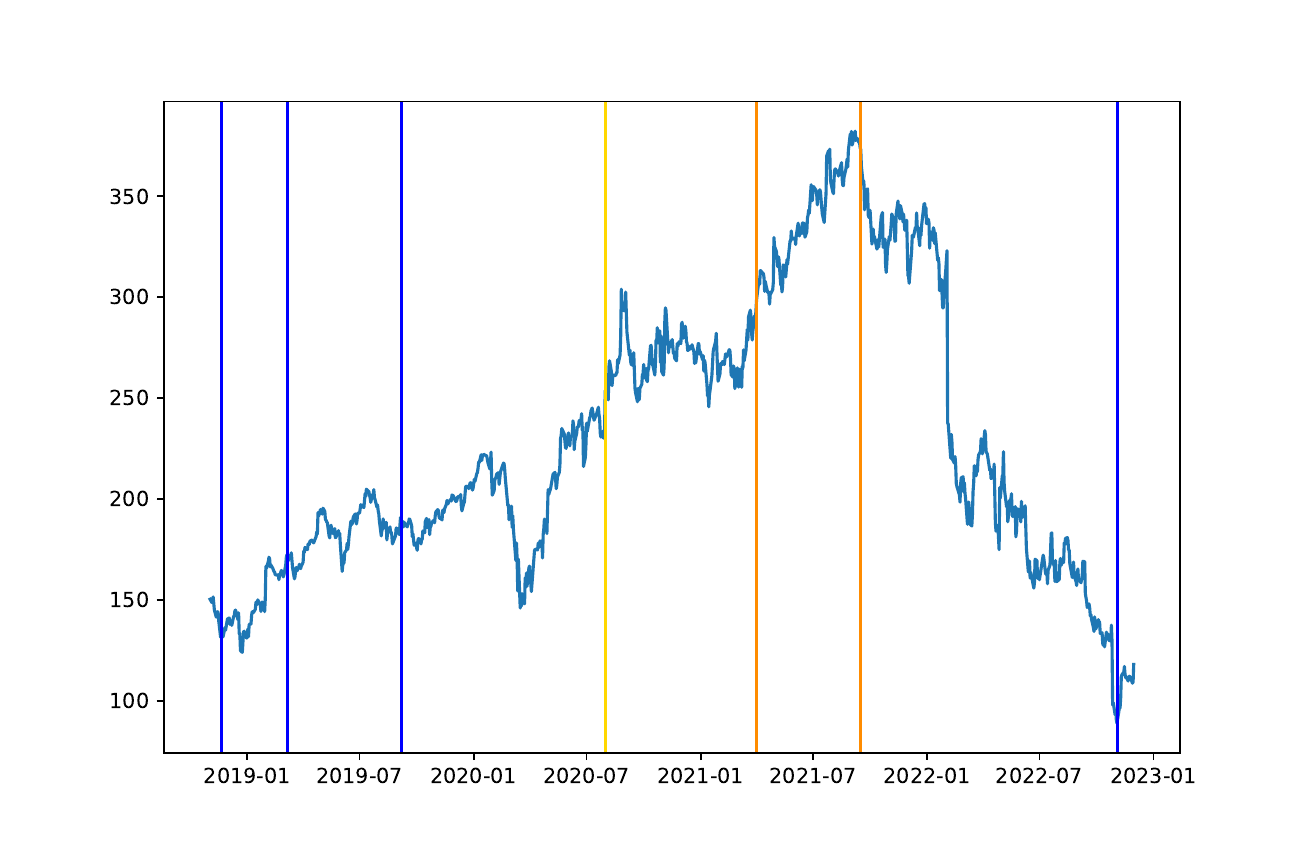}
 \caption{\black{Meta Platforms Inc. stock prices (USD). The lines illustrate the dates considered to detect a bubble with our methodology. The colors change based on the average probability $P_b$ that the asset has a bubble according to our methodology. In particular, for the orange lines we have \black{$92<P_b<94\%$}, for the  yellow line $P_b=82.1\%$ and for the blue ones $P_b<=5\%$.}}
 \label{meta}
\end{figure}
\begin{figure}
\centering
 \includegraphics[scale=0.9]{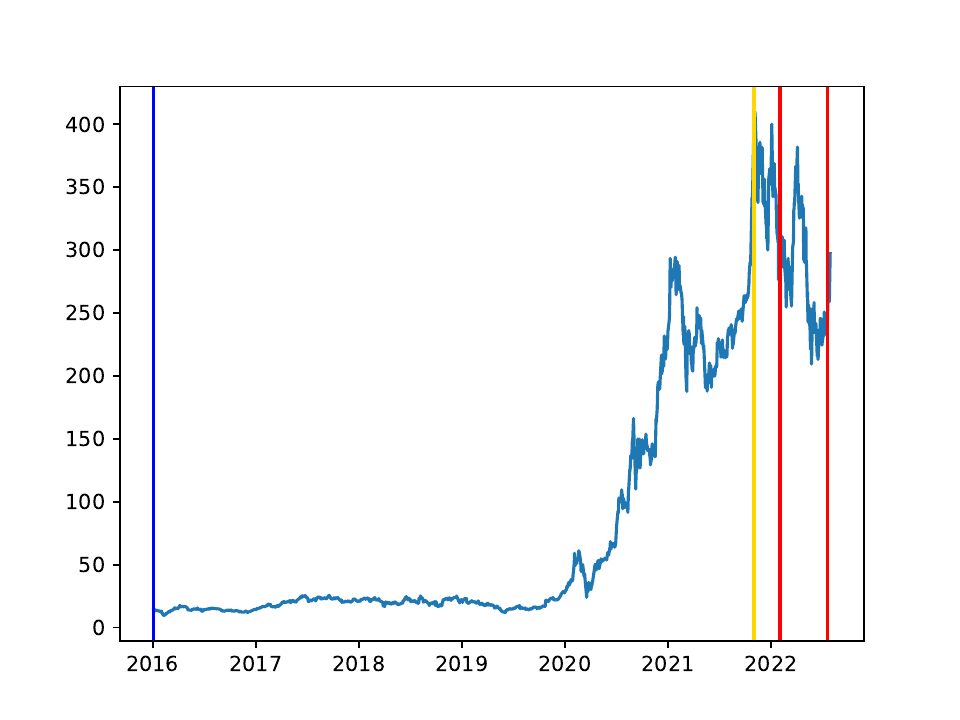}
 \caption{Tesla Inc. stock prices (USD). The lines illustrate the dates considered to detect a bubble with our methodology. The colors change based on the \black{average} probability $P_b$ that the asset has a bubble according to our methodology. In particular: \black{for the red lines we have $P_b\ge95\%$, for the yellow line $P_b=70.1\%$ and for the blue one $P_b=1.3\%$}.}
 \label{tesla}
\end{figure}

\section{Conclusion}\label{sec:conclusions}

In this article we propose a deep learning-based methodology to detect financial asset price bubbles, where %The methodology can be summarized as follows. 
a deep neural network is trained using synthetic option price data generated from a collection of models. \black{We apply our method to} local and stochastic volatility models, in which the presence of a bubble is determined by the value of certain parameters. This trained neural network is then used for bubble detection in more general situations. 
The deep neural network takes as input call option prices for a given underlying and returns as output a number in \blue{$(0,1)$}, which indicates the probability that the underlying exhibits a bubble. 

We have assessed the performance of the methodology on both synthetic and real data. First we have tested various settings in which the network is trained using a certain class of local or stochastic volatility models and tested within another class. In particular, a high out-of-sample prediction accuracy has been observed when two model classes are used for training and the trained network is used for bubble detection in another class of models. We have also tested
the methodology with market data of \black{four} tech stocks at various dates and \black{have observed a close match between output of the network and expected results}.
Furthermore, these numerical experiments have been complemented by a mathematical proof for the method in a local volatility setting. Extensions of these theoretical results to more general model classes are also possible, provided an analogue of the ``bubble detection function'' $F$ can be constructed. 
On the other hand, the proposed algorithm could be directly modified to take into account additional market factors as inputs for the neural network, provided that these factors can be incorporated into the models used to generate the training data.

\bibliography{bib}

\begin{thebibliography}{52}
\providecommand{\natexlab}[1]{#1}
\providecommand{\url}[1]{\texttt{#1}}
\expandafter\ifx\csname urlstyle\endcsname\relax
  \providecommand{\doi}[1]{doi: #1}\else
  \providecommand{\doi}{doi: \begingroup \urlstyle{rm}\Url}\fi

\bibitem[Abreu and Brunnermeier(2003)]{Abreu}
D.~Abreu and M.K. Brunnermeier.
\newblock Bubbles and crashes.
\newblock \emph{Econometrica}, 71\penalty0 (1):\penalty0 173--204, 2003.

\bibitem[Barron(1993)]{Barron1993}
Andrew~R. Barron.
\newblock Universal approximation bounds for superpositions of a sigmoidal
  function.
\newblock \emph{IEEE Trans. Inform. Theory}, 39\penalty0 (3):\penalty0
  930--945, 1993.

\bibitem[Bercovici(2017)]{Bercovici}
Jeff Bercovici.
\newblock Yes, it's a {Tech} {Bubble}.
\newblock \emph{Inc.}, 2017.
\newblock URL
  \url{https://www.inc.com/magazine/201509/jeff-bercovici/are-we-in-a-tech-bubble.html}.

\bibitem[Biagini and Nedelcu(2015)]{BiaginiNedelcu}
F.~Biagini and S.~Nedelcu.
\newblock {The formation of financial bubbles in defaultable markets}.
\newblock \emph{SIAM Journal on Financial Mathematics}, 6\penalty0
  (1):\penalty0 530--558, 2015.

\bibitem[Biagini et~al.(2018)Biagini, Mazzon, and
  Meyer-Brandis]{biagini_maz_mb}
F.~Biagini, A.~Mazzon, and T.~Meyer-Brandis.
\newblock Liquidity induced asset bubbles via flows of {ELMM}s.
\newblock \emph{SIAM Journal on Financial Mathematics}, 9\penalty0
  (2):\penalty0 800--834, 2018.

\bibitem[Biagini et~al.(2014)Biagini, F{\"o}llmer, and Nedelcu]{Biagini}
Francesca Biagini, Hans F{\"o}llmer, and Sorin Nedelcu.
\newblock {Shifting martingale measures and the slow birth of a bubble}.
\newblock \emph{Finance and Stochastics}, 18\penalty0 (2):\penalty0 297--326,
  2014.

\bibitem[Carter et~al.(2016)Carter, Pan, Rai, and Galandiuk]{carter2016roc}
Jane~V Carter, Jianmin Pan, Shesh~N Rai, and Susan Galandiuk.
\newblock Roc-ing along: Evaluation and interpretation of receiver operating
  characteristic curves.
\newblock \emph{Surgery}, 159\penalty0 (6):\penalty0 1638--1645, 2016.

\bibitem[Cox and Hobson(2005)]{cox2005local}
Alexander~MG Cox and David~G Hobson.
\newblock Local martingales, bubbles and option prices.
\newblock \emph{Finance and Stochastics}, 9\penalty0 (4):\penalty0 477--492,
  2005.

\bibitem[DeLong et~al.(1990)DeLong, Shleifer, Summers, and Waldmann]{DeLong}
J.B. DeLong, A.~Shleifer, L.~Summers, and R.~Waldmann.
\newblock Noise trader risk in financial markets.
\newblock \emph{Journal of Political Economy}, 98\penalty0 (4):\penalty0
  703--738, 1990.

\bibitem[Dias et~al.(2020)Dias, Nunes, and Cruz]{dias2020note}
Jos{\'e}~Carlos Dias, Jo{\~a}o Pedro~Vidal Nunes, and Aricson Cruz.
\newblock A note on options and bubbles under the {CEV} model: implications for
  pricing and hedging.
\newblock \emph{Review of Derivatives Research}, 23\penalty0 (3):\penalty0
  249--272, 2020.

\bibitem[Ekstr{\"o}m and Tysk(2012)]{ekstrom2012dupire}
Erik Ekstr{\"o}m and Johan Tysk.
\newblock Dupire's equation for bubbles.
\newblock \emph{International Journal of Theoretical and Applied Finance},
  15\penalty0 (06):\penalty0 1250041, 2012.

\bibitem[Elbr\"{a}chter et~al.(2022)Elbr\"{a}chter, Grohs, Jentzen, and
  Schwab]{EGJS18_787}
Dennis Elbr\"{a}chter, Philipp Grohs, Arnulf Jentzen, and Christoph Schwab.
\newblock D{NN} expression rate analysis of high-dimensional {PDE}s:
  application to option pricing.
\newblock \emph{Constr. Approx.}, 55\penalty0 (1):\penalty0 3--71, 2022.

\bibitem[Fan et~al.(2006)Fan, Upadhye, and Worster]{fan2006understanding}
Jerome Fan, Suneel Upadhye, and Andrew Worster.
\newblock Understanding receiver operating characteristic (roc) curves.
\newblock \emph{Canadian Journal of Emergency Medicine}, 8\penalty0
  (1):\penalty0 19--20, 2006.

\bibitem[F\"ollmer et~al.(2005)F\"ollmer, Horst, and Kirman]{foellmer}
H.~F\"ollmer, U.~Horst, and A.~Kirman.
\newblock Equilibria in financial markets with heterogeneous agents: A
  probabilistic perspective.
\newblock \emph{Journal of Mathematical Economics}, 41\penalty0 (1-2):\penalty0
  123--155, 2005.

\bibitem[Fusari et~al.(2020)Fusari, Jarrow, and Lamichhane]{fusari2020testing}
Nicola Fusari, Robert~A. Jarrow, and Sujan Lamichhane.
\newblock Testing for asset price bubbles using options data.
\newblock \emph{Johns Hopkins Carey Business School Research Paper}, \penalty0
  (20-12), 2020.

\bibitem[Gon{\c{c}}alves et~al.(2014)Gon{\c{c}}alves, Subtil, Oliveira, and
  de~Zea~Bermudez]{gonccalves2014roc}
Luzia Gon{\c{c}}alves, Ana Subtil, M~Ros{\'a}rio Oliveira, and Patricia
  de~Zea~Bermudez.
\newblock Roc curve estimation: An overview.
\newblock \emph{REVSTAT-Statistical journal}, 12\penalty0 (1):\penalty0 1--20,
  2014.

\bibitem[Gonon(2021)]{Gonon2021}
Lukas Gonon.
\newblock {Random feature neural networks learn Black-Scholes type PDEs without
  curse of dimensionality}.
\newblock \emph{Preprint, arXiv 2106.08900}, 2021.

\bibitem[Gonon and Schwab(2021)]{GS20_925}
Lukas Gonon and Christoph Schwab.
\newblock Deep {R}e{LU} network expression rates for option prices in
  high-dimensional, exponential {L}\'{e}vy models.
\newblock \emph{Finance Stoch.}, 25\penalty0 (4):\penalty0 615--657, 2021.

\bibitem[Grohs et~al.(2018)Grohs, Hornung, Jentzen, and von
  Wurstemberger]{HornungJentzen2018}
Philipp Grohs, Fabian Hornung, Arnulf Jentzen, and Philippe von Wurstemberger.
\newblock A proof that artificial neural networks overcome the curse of
  dimensionality in the numerical approximation of {B}lack-{S}choles partial
  differential equations.
\newblock \emph{To appear in Mem. Am. Math. Soc.; arXiv:1809.02362}, 2018.

\bibitem[Harrison and Kreps(1978)]{hk}
J.M. Harrison and D.M. Kreps.
\newblock Speculative investor behavior in a stock market with heterogeneous
  expectations.
\newblock \emph{The Quarterly Journal of Economics}, 92\penalty0 (2):\penalty0
  323--336, 1978.

\bibitem[Herdegen and Schweizer(2016)]{herdegen2016strong}
Martin Herdegen and Martin Schweizer.
\newblock Strong bubbles and strict local martingales.
\newblock \emph{International Journal of Theoretical and Applied Finance},
  19\penalty0 (04):\penalty0 1650022, 2016.

\bibitem[Heston et~al.(2007)Heston, Loewenstein, and
  Willard]{heston2007options}
Steven~L Heston, Mark Loewenstein, and Gregory~A Willard.
\newblock Options and bubbles.
\newblock \emph{The Review of Financial Studies}, 20\penalty0 (2):\penalty0
  359--390, 2007.

\bibitem[Hornik(1991)]{hornik1991}
Kurt Hornik.
\newblock {Approximation capabilities of multilayer feedforward networks}.
\newblock \emph{Neural Networks}, 4\penalty0 (1989):\penalty0 251--257, 1991.

\bibitem[Jacquier and Keller-Ressel(2018)]{jacquier2018implied}
Antoine Jacquier and Martin Keller-Ressel.
\newblock Implied volatility in strict local martingale models.
\newblock \emph{SIAM Journal on Financial Mathematics}, 9\penalty0
  (1):\penalty0 171--189, 2018.

\bibitem[Jarrow and Madan(2000)]{JarrowMadan}
R.~Jarrow and D.~Madan.
\newblock Arbitrage, martingales and private monetary value.
\newblock \emph{Journal of Risk}, 3\penalty0 (1), 2000.

\bibitem[Jarrow and Protter(2011)]{JarrowProtter2011}
R.~Jarrow and P.~Protter.
\newblock Foreign currency bubbles.
\newblock \emph{{ Review of Derivatives Research}}, 14\penalty0 (1):\penalty0
  67--83, 2011.

\bibitem[Jarrow and Protter(2013)]{prrel}
R.~Jarrow and P.~Protter.
\newblock Relative asset price bubbles.
\newblock Preprint, Available at SSRN: http://ssrn.com/abstract=2265465 or
  http://dx.doi.org/10.2139/ssrn.2265465, 2013.

\bibitem[Jarrow et~al.(2012)Jarrow, Protter, and Roch]{JarrowProtter2012}
R.~Jarrow, P.~Protter, and A.~Roch.
\newblock { A Liquidity Based Model for Asset Price}.
\newblock \emph{Quantitative Finance}, 12\penalty0 (1):\penalty0 1339--1349,
  2012.

\bibitem[Jarrow(2016)]{jarrow2016testing}
Robert~A. Jarrow.
\newblock Testing for asset price bubbles: three new approaches.
\newblock \emph{Quantitative Finance Letters}, 4\penalty0 (1):\penalty0 4--9,
  2016.

\bibitem[Jarrow and Kwok(2021)]{jarrow2021inferring}
Robert~A. Jarrow and Simon~S. Kwok.
\newblock Inferring financial bubbles from option data.
\newblock \emph{Journal of Applied Econometrics}, 36\penalty0 (7):\penalty0
  1013--1046, 2021.

\bibitem[Jarrow and Protter(2009)]{JarrowProtter2009}
Robert~A. Jarrow and Philip Protter.
\newblock { Forward and futures prices with bubbles}.
\newblock \emph{International Journal of Theoretical and Applied Finance},
  12\penalty0 (7):\penalty0 901--924, 2009.

\bibitem[Jarrow et~al.(2007)Jarrow, Protter, and Shimbo]{JarrowProtter2007}
Robert~A. Jarrow, Philip Protter, and Kazuhiro Shimbo.
\newblock {Asset price bubbles in complete markets}.
\newblock \emph{Advances in Mathematical Finance}, In Honor of Dilip B.
  Madan:\penalty0 105--130, 2007.

\bibitem[Jarrow et~al.(2010)Jarrow, Protter, and Shimbo]{JarrowProtter2010}
Robert~A. Jarrow, Philip Protter, and Kazuhiro Shimbo.
\newblock {Asset price bubbles in incomplete markets}.
\newblock \emph{Mathematical Finance}, 20(2):\penalty0 145--185, 2010.

\bibitem[Jarrow et~al.(2011{\natexlab{a}})Jarrow, Kchia, and
  Protter]{JarrowKchiaProtter}
Robert~A. Jarrow, Younes Kchia, and Philip Protter.
\newblock { How to detect an asset bubble}.
\newblock \emph{SIAM Journal on Financial Mathematics}, 2:\penalty0 839--865,
  2011{\natexlab{a}}.

\bibitem[Jarrow et~al.(2011{\natexlab{b}})Jarrow, Kchia, and
  Protter]{jarrow2011detect}
Robert~A. Jarrow, Younes Kchia, and Philip Protter.
\newblock How to detect an asset bubble.
\newblock \emph{SIAM Journal on Financial Mathematics}, 2\penalty0
  (1):\penalty0 839--865, 2011{\natexlab{b}}.

\bibitem[Jarrow et~al.(2011{\natexlab{c}})Jarrow, Kchia, and
  Protter]{jarrow2011there}
Robert~A. Jarrow, Younes Kchia, and Philip Protter.
\newblock Is there a bubble in {L}inkedin's stock price?
\newblock \emph{The Journal of Portfolio Management}, 38\penalty0 (1):\penalty0
  125--130, 2011{\natexlab{c}}.

\bibitem[Kardaras et~al.(2015)Kardaras, Kreher, and
  Nikeghbali]{KardarasKreherNikeghbali}
C.~Kardaras, D.~Kreher, and A.~Nikeghbali.
\newblock Strict local martingales and bubbles.
\newblock \emph{The Annals of Applied Probability}, 25\penalty0 (4):\penalty0
  1827--1867, 2015.

\bibitem[Kolakowski(2019)]{kolakowski}
Mark Kolakowski.
\newblock Nvidia's stock signals techs near bubble like 2000.
\newblock 2019.
\newblock URL \url{https://www.investopedia.com/news/nvidia/}.

\bibitem[Libich et~al.(2021)Libich, Lenten, et~al.]{libich2021bitcoin}
Jan Libich, Liam Lenten, et~al.
\newblock Bitcoin, {T}esla and {G}ame{S}top bubbles as a flight to focal
  points.
\newblock \emph{World Economics}, 22\penalty0 (1):\penalty0 83--108, 2021.

\bibitem[Lindsay and Brecher(2012)]{lindsay2012simulation}
Alan~E Lindsay and DR~Brecher.
\newblock Simulation of the {CEV} process and the local martingale property.
\newblock \emph{Mathematics and Computers in Simulation}, 82\penalty0
  (5):\penalty0 868--878, 2012.

\bibitem[Loewenstein and Willard(2000)]{LoewensteinWillard}
Mark Loewenstein and Gregory~A. Willard.
\newblock Rational equilibrium asset-pricing bubbles in continuous trading
  models.
\newblock \emph{Journal of Economic Theory}, 91(1):\penalty0 17--58, 2000.

\bibitem[Ozimek(2017)]{Ozimek}
Adam Ozimek.
\newblock There obviosly is a tech bubble, but hopefully that doesn't matter.
\newblock \emph{Forbes}, 2017.
\newblock URL
  \url{https://www.forbes.com/sites/modeledbehavior/2017/07/16/there-obviously-is-a-tech-bubble-but-hopefully-that-doesnt-matter/#737270d9d426}.

\bibitem[Pal and Protter(2010)]{PalProtter}
S.~Pal and P.~Protter.
\newblock Analysis of continuous strict local martingales via h-transforms.
\newblock \emph{Stochastic processes and their applications}, 120\penalty0
  (8):\penalty0 1424--1443, 2010.

\bibitem[Piiroinen et~al.(2018)Piiroinen, Roininen, Schoden, and
  Simon]{piiroinen2018asset}
Petteri Piiroinen, Lassi Roininen, Tobias Schoden, and Martin Simon.
\newblock Asset price bubbles: An option-based indicator.
\newblock \emph{arXiv preprint arXiv:1805.07403}, 2018.

\bibitem[Protter(2013)]{Protter2013}
P.~Protter.
\newblock A mathematical theory of financial bubbles.
\newblock In V.~Henderson and R.~Sincar, editors, \emph{Paris-Princeton
  Lectures on Mathematical Finance}, volume 2081 of \emph{Lecture Notes in
  Mathematics}. Springer, 2013.

\bibitem[Reisinger and Zhang(2020)]{ReisingerZhang2019}
Christoph Reisinger and Yufei Zhang.
\newblock Rectified deep neural networks overcome the curse of dimensionality
  for nonsmooth value functions in zero-sum games of nonlinear stiff systems.
\newblock \emph{Anal. Appl. (Singap.)}, 18\penalty0 (6):\penalty0 951--999,
  2020.

\bibitem[Revuz and Yor(2013)]{revuz2013continuous}
Daniel Revuz and Marc Yor.
\newblock \emph{Continuous martingales and Brownian motion}, volume 293.
\newblock Springer Science \& Business Media, 2013.

\bibitem[Scheinkman and Xiong(2003)]{Scheinkman}
J.~Scheinkman and W.~Xiong.
\newblock {Overconfidence and speculative bubbles}.
\newblock \emph{Journal of political economy}, 111\penalty0 (6):\penalty0
  1183--1219, 2003.

\bibitem[Serla(2017)]{Seria}
Rusli Serla.
\newblock The tech bubble: how close is it to bursting?
\newblock \emph{Telegraph}, 2017.
\newblock URL
  \url{http://www.telegraph.co.uk/technology/2017/06/22/tech-bubble-close-bursting/}.

\bibitem[Sharma(2017)]{Sharma}
Ruchir Sharma.
\newblock When will the tech bubble burst?
\newblock \emph{New York Times}, 2017.
\newblock URL
  \url{https://www.nytimes.com/2017/08/05/opinion/sunday/when-will-the-tech-bubble-burst.html}.

\bibitem[Sin(1998)]{sin1998complications}
Carlos~A. Sin.
\newblock Complications with stochastic volatility models.
\newblock \emph{Advances in Applied Probability}, 30\penalty0 (1):\penalty0
  256--268, 1998.

\bibitem[Tirole(1982)]{tirole}
J.~Tirole.
\newblock On the possibility of speculation under rational expectations.
\newblock \emph{Econometrica}, 53\penalty0 (6):\penalty0 1163--1182, 1982.

\end{thebibliography}
\bibliographystyle{plainnat}

\section{Appendix}
We here provide the ROC curves relative to the experiments of Section \ref{sec:stochvol}.

\black{\begin{figure}[h]
\subfloat{\includegraphics[width = 2.5in]{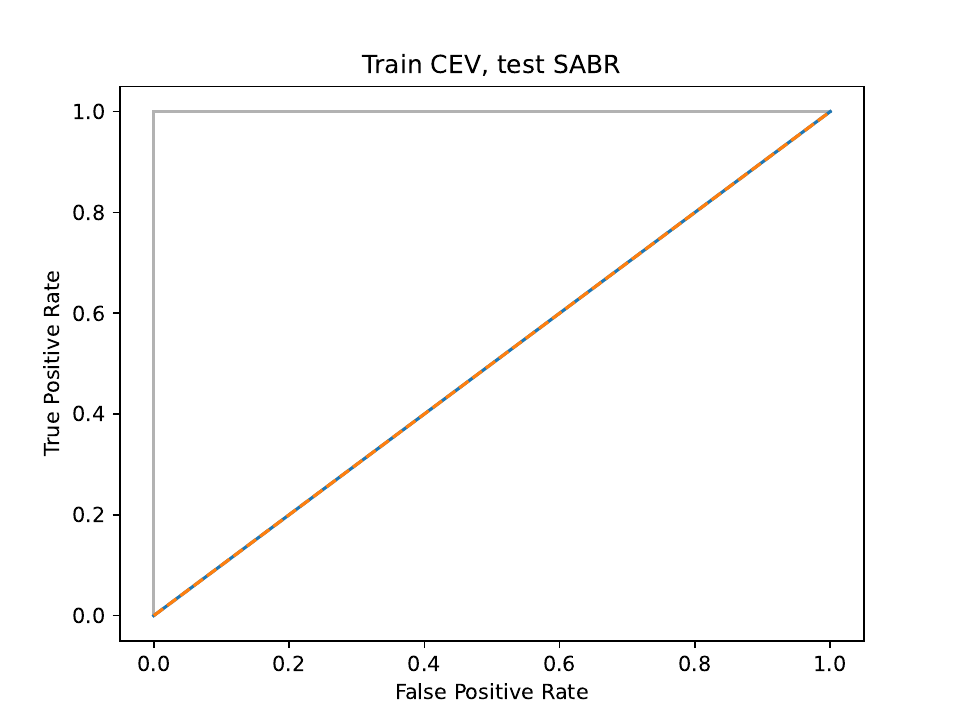}}
\subfloat{\includegraphics[width = 2.5in]{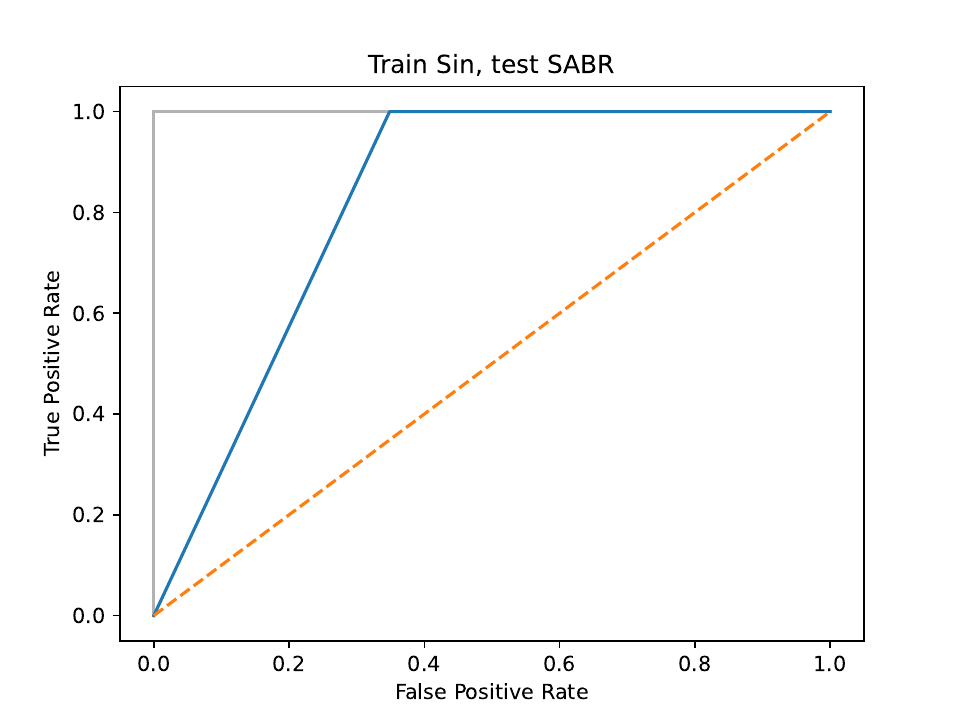}}\\
\subfloat{\includegraphics[width = 2.5in]{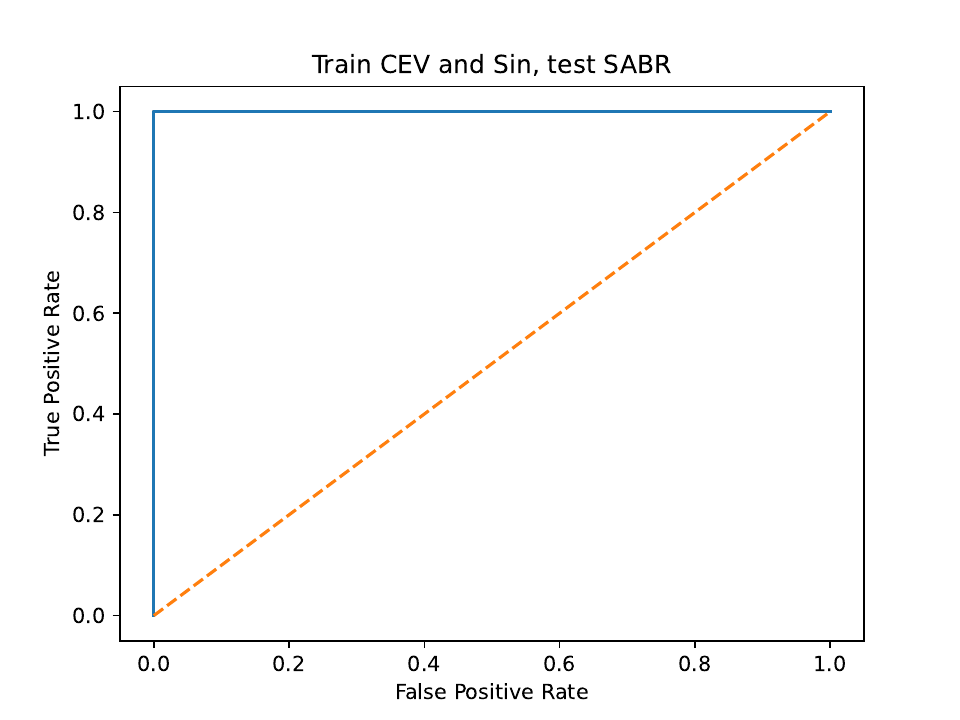}}
\subfloat{\includegraphics[width = 2.5in]{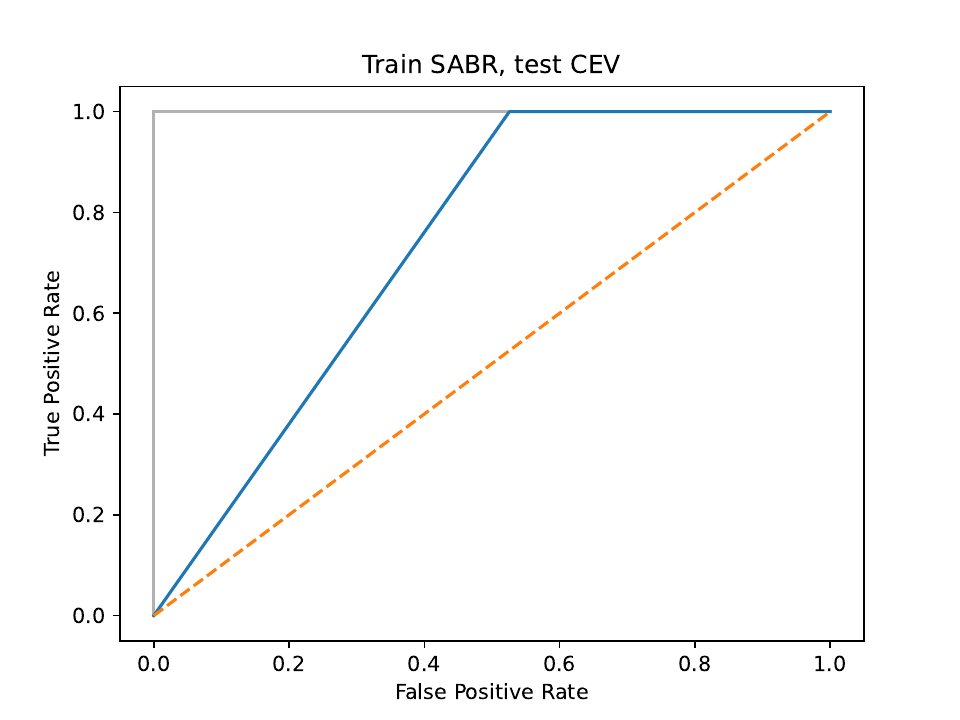}} \\
\subfloat{\includegraphics[width = 2.5in]{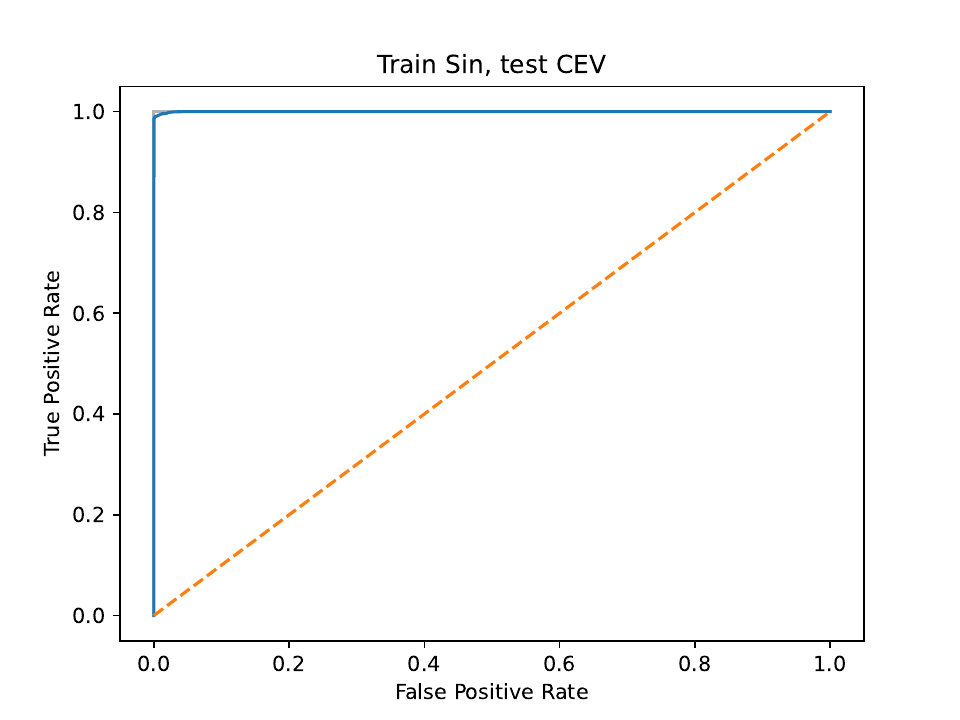}}
\subfloat{\includegraphics[width = 2.5in]{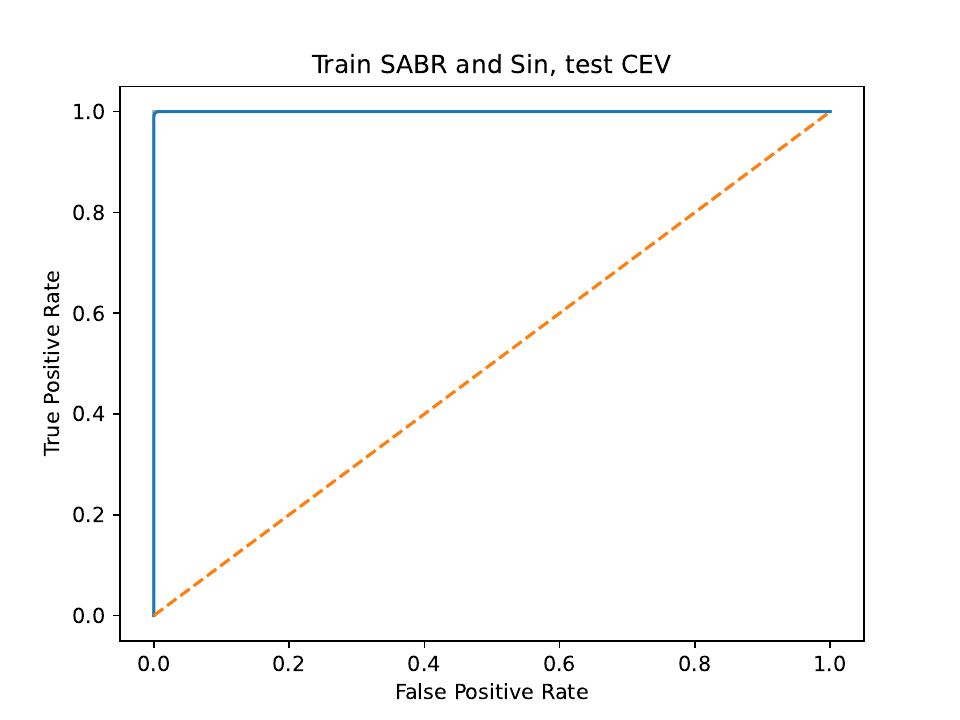}} 
\caption{ROC curves obtained taking underlyings with initial value $x_0=2$ and dynamics specified in Section \ref{sec:stochvol}, equally spaced maturities in $(2,5)$ and equally spaced strikes in $(1,3.5)$. }
\label{fig:rocat}
\end{figure}
\begin{figure}
\subfloat{\includegraphics[width = 2.5in]{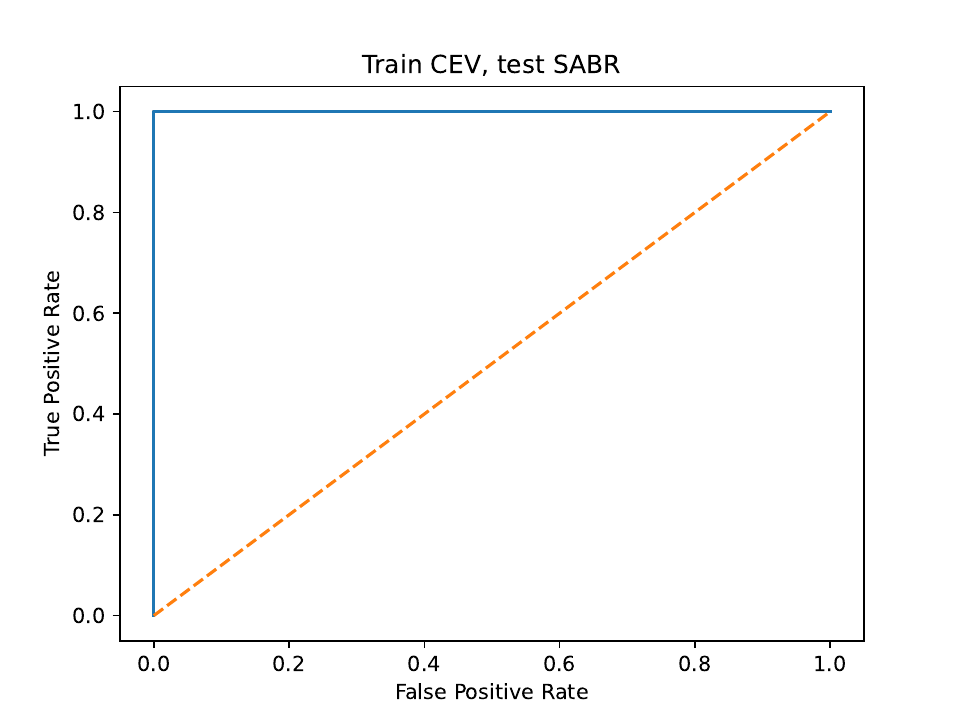}}
\subfloat{\includegraphics[width = 2.5in]{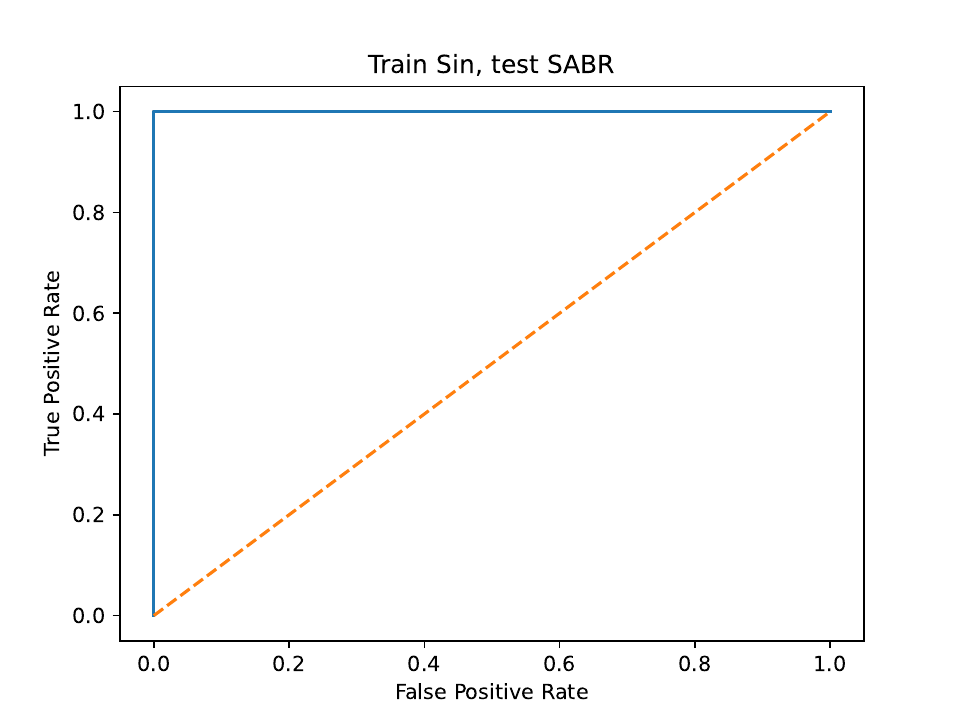}}\\
\subfloat{\includegraphics[width = 2.5in]{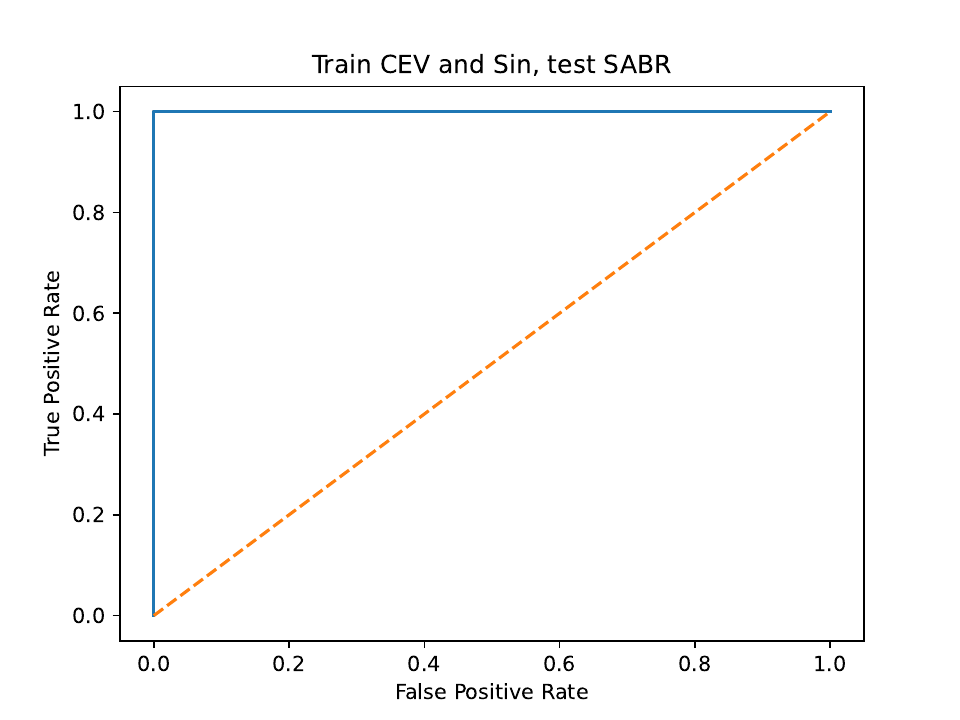}}
\subfloat{\includegraphics[width = 2.5in]{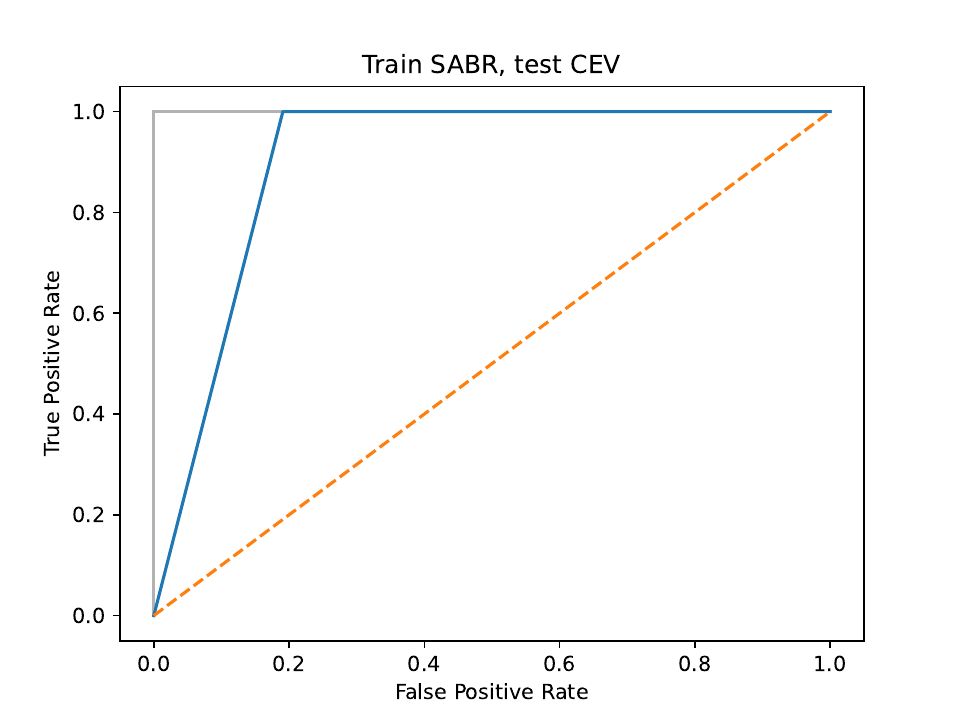}} \\
\subfloat{\includegraphics[width = 2.5in]{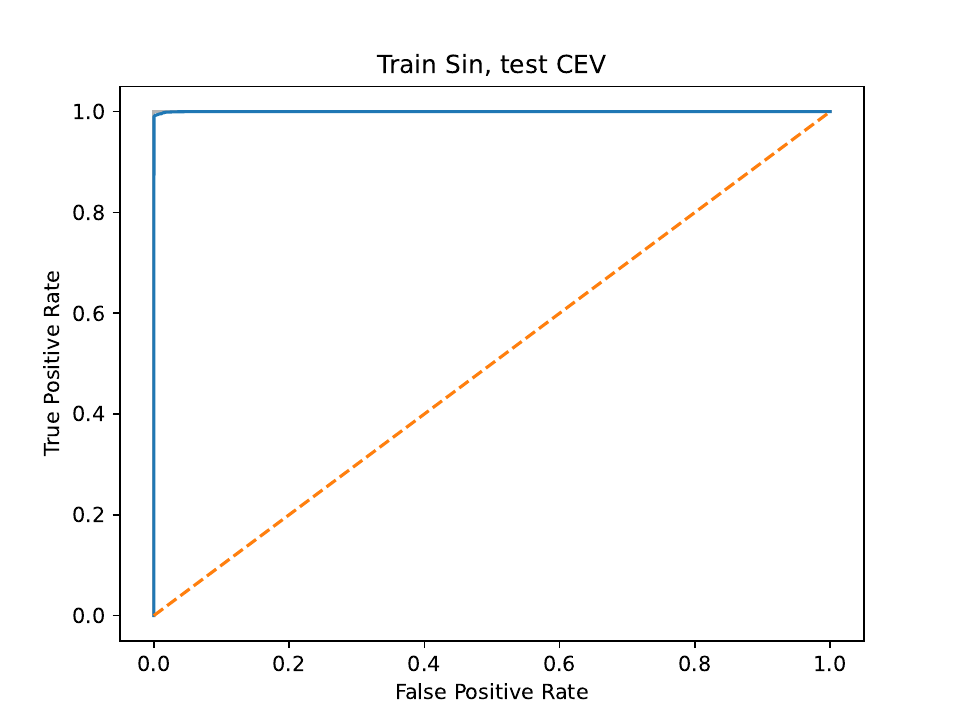}}
\subfloat{\includegraphics[width = 2.5in]{roc_train_sabr_sin_test_cev.pdf}} 
\caption{ROC curves obtained taking underlyings with initial value $x_0=2.0$ and dynamics specified in Section \ref{sec:stochvol}, equally spaced maturities in $(2,5)$ and equally spaced strikes in $(3,5.5)$. }
\label{fig:rocout}
\end{figure}}

\end{document}